\documentclass[letterpaper, 11pt]{article}

\usepackage{amsmath,amssymb,amsthm}
\usepackage[american]{babel}
\usepackage{balance} %
\usepackage{enumitem} %
\usepackage{float} %
\usepackage[T1]{fontenc}
\usepackage{fullpage} %
\usepackage{graphicx}
\usepackage[utf8]{inputenc}
\usepackage{mathtools} %
\usepackage{mdframed}
\usepackage{multirow} %
\usepackage{tikz}
\usepackage{titling}
\usepackage{thmtools}
\usepackage{url}
\usepackage{xcolor}
\usepackage[backref=page, pagebackref=true, linktocpage]{hyperref}
\usepackage{algorithm}
\usepackage[noend]{algpseudocode}
\usepackage{cleveref} %
\usepackage{comment}
\usepackage{bm}

\MakeRobust{\Call}
\makeatletter
\newcounter{algorithmicH}%
\let\oldalgorithmic\algorithmic
\renewcommand{\algorithmic}{%
  \stepcounter{algorithmicH}%
  \oldalgorithmic}%
\renewcommand{\theHALG@line}{ALG@line.\thealgorithmicH.\arabic{ALG@line}}
\makeatother
\algnewcommand\Break{\textbf{break}}
\algnewcommand\Continue{\textbf{continue}}
\algnewcommand\Exit{\textbf{exit}}
\algnewcommand\Or{\textbf{or~}}
\algnewcommand\And{\textbf{and~}}
\algnewcommand{\IfThenElse}[3]{\State \algorithmicif\ #1\ \algorithmicthen\ #2\ \algorithmicelse\ #3}
\algnewcommand{\IfThenNLElse}[3]{\State \algorithmicif\ #1\ \algorithmicthen\ #2\ \State\algorithmicelse\ #3}
\algnewcommand{\IfThen}[2]{\State \algorithmicif\ #1\ \algorithmicthen\ #2}
\algnewcommand{\ForInline}[2]{\State \algorithmicfor\ #1\ \algorithmicdo\ #2}

\DeclareMathOperator{\polylog}{polylog}
\DeclareMathOperator{\bas}{base}

\DeclareMathOperator{\lab}{label}
\DeclareMathOperator{\tab}{table}
\DeclareMathOperator{\port}{port}
\DeclareMathOperator{\border}{border}

\newcommand{\shay}[1]{\textcolor{red}{[[S: #1]]}}

\newcommand{\eps}{\epsilon}
\newcommand{\pruned}{\mathsf{pnd}}

\newcommand{\ignore}[1]{}
\newcommand{\dst}{\mathsf{dist}}

\newcommand{\diam}{\mathsf{diam}}

\crefname{claim}{claim}{claims}
\Crefname{claim}{Claim}{Claims}

\crefname{property}{prperty}{properties}
\Crefname{property}{Property}{Properties}

\definecolor{redlink}{rgb}{0.6, 0, 0}
\definecolor{greenlink}{rgb}{0, 0.6, 0}
\definecolor{bluelink}{rgb}{0, 0, 0.6}
\hypersetup{
	colorlinks=true,
	linkcolor=redlink, %
	urlcolor=redlink, %
	citecolor=greenlink, %
}

\theoremstyle{definition}
\newtheorem{definition}{Definition}[section]
\newtheorem{question}{Question}[section]
\newtheorem{property}{Property}

\theoremstyle{plain}
\newtheorem{theorem}{Theorem}[section]
\newtheorem{lemma}{Lemma}[section]

\newtheorem{observation}{Observation}[section]

\theoremstyle{remark}
\newtheorem{remark}{Remark}[section]
\newtheorem{claim}{Claim}[section]

\DeclarePairedDelimiter\lAngle{\langle}{\rangle}

\author{
Omri Kahalon\thanks{Tel Aviv University. Email: \href{mailto:kahalon@gmail.com}{kahalon@gmail.com}.}
\and
Hung Le\thanks{University of Massachusetts Amherst. Email: \href{mailto:hungle@cs.umass.edu}{hungle@cs.umass.edu}.}
\and
Lazar Milenković\thanks{Tel Aviv University. Email: \href{mailto:milenkovic.lazar@gmail.com}{milenkovic.lazar@gmail.com}.}
\and
Shay Solomon\thanks{Tel Aviv University. Email: \href{mailto:solo.shay@gmail.com}{solo.shay@gmail.com}.}}

\title{Can't See The Forest for the Trees: Navigating Metric Spaces by Bounded Hop-Diameter Spanners}
\predate{}\date{}\postdate{}

\begin{document}

\maketitle
\begin{abstract}
Spanners for metric spaces
have been extensively studied, perhaps most notably in low-dimensional Euclidean spaces
 --- due to their numerous applications.  
{\em Euclidean spanners} can be viewed as means of {\em compressing} the $\binom{n}{2}$ pairwise distances of a $d$-dimensional Euclidean space into $O(n) = O_{\epsilon,d}(n)$ spanner edges, so that the spanner distances preserve the original distances to within a factor of $1+\epsilon$, for any $\epsilon > 0$. Moreover, one can compute such spanners efficiently in the standard centralized and distributed settings. Once the spanner has been computed, it serves as a ``proxy'' overlay network, on which the computation can proceed, which gives rise to huge savings in space and other important quality measures.

The original metric enables us to ``navigate'' optimally --- a single hop (for any two points) with the exact distance,
but the price is high --- $\Theta(n^2)$ edges.
Is it possible to efficiently navigate, on a sparse spanner, using $k$ hops and approximate distances, for $k$ close to 1 (say $k=2$)? Surprisingly, this fundamental question has been overlooked in Euclidean spaces, as well as in other classes of metrics, despite the long line of work on spanners in metric spaces.

We answer this question in the affirmative via a surprisingly simple observation on bounded hop-diameter spanners for {\em tree metrics}, which we apply on top of known, as well as new, {\em tree cover theorems}. Beyond its simplicity, the strength of our approach is three-fold:
\begin{itemize}
\item {\bf Applicable}: We present a variety of applications of our efficient navigation scheme,
including a 2-hop routing scheme in Euclidean spaces with stretch $1+\epsilon$
using $O(\log^2 n)$ bits of memory for labels and routing tables --- to the best of our knowledge, all known routing schemes prior to this work use $\Omega(\log n)$ hops.
\item {\bf Unified}: Our navigation scheme and applications extend beyond Euclidean spaces to any class of metrics that admits an efficient tree cover theorem; currently this includes doubling, planar and general metrics, but our approach is unified.
\item {\bf Fault-Tolerant}: In Euclidean and doubling metrics, we strengthen all our results to achieve fault-tolerance. 
To this end, we first design a new construction of fault-tolerant spanners of bounded hop-diameter, which, in turn, relies on a new tree cover theorem for doubling metrics --- hereafter the ``Robust Tree Cover'' Theorem, which generalizes the classic ``Dumbbell Tree'' Theorem [Arya et al., STOC'95] in Euclidean spaces. 
\end{itemize}
\end{abstract}
\thispagestyle{empty}
\pagebreak

\tableofcontents
\clearpage
\section{Introduction}
\subsection{Background and motivation}
Let $M_X=(X, \delta_X)$ be an $n$-point metric space, viewed as a complete weighted graph whose weight function satisfies the triangle inequality.  %
For a parameter $t \ge 1$, a subgraph $H = (V,E',w)$ of $M_X$
($E' \subseteq \binom{V}{2}$) 
is called a \emph{$t$-spanner} for $M_X$ if for all $u,v \in V$, $\delta_H(u,v) \le t \cdot \delta_X(u,v)$.
(Here $\delta_X(u,v)$ and $\delta_H(u,v)$ denote the distances between $u$ and $v$ in $M_X$ and the spanner $H$, respectively.)
In other words, for all $u,v \in V$, there exists a path in $H$ between $u$ and $v$
whose {\em weight} (sum of edge weights in it) is at most $t \cdot \delta_X(u,v)$;
such a path is called a {\em $t$-spanner path} and the parameter $t$ is called the \emph{stretch} of $H$. 
Since their introduction in the late 80s \cite{PS89,PU89}, spanners have been extensively studied,
and by now they are recognized as a graph structure of fundamental importance, in both theory and practice. 

There are a few basic properties of spanners that are important for a wide variety of practical applications; in most applications,  a subset of these properties need to be satisfied while preserving small stretch.
Although the exact subset of properties varies between applications, perhaps the most basic property (besides small stretch) is to have a small number of edges (or \emph{size}), close to $O(n)$; the spanner \emph{sparsity} is the ratio of its size and the size $n-1$ of a spanning tree. %
Second, the spanner \emph{weight} $w(H) \coloneqq \sum_{e \in E'} w(e)$ should be close to the weight $w(MST(M_X))$ of a minimum spanning tree $MST(M_X)$ of the underlying metric; we refer to the normalized notion of weight, $w(H)/w(MST(M_X))$, as the spanner \emph{lightness}.
Third, the \emph{hop-diameter} of a spanner should be close to 1; the hop-diameter of a $t$-spanner
is the smallest integer $k$ such that for all $u, v \in V$, there exists a $t$-spanner path between $u$ and $v$ with at most $k$ edges (or {\em hops}).
Finally, the \emph{degree} of a spanner, 
i.e., the maximum number of edges incident on any vertex, should be close to constant.

The original motivation of spanners was in distributed computing. For example, light and sparse spanners have been used in reducing the communication cost in efficient broadcast protocols~\cite{ABP90,ABP92}, synchronizing networks and computing global functions \cite{Awerbuch85,PU89,Peleg00}, gathering and disseminating data~\cite{BKRCV02,VWFME03,KV01}, and routing~\cite{WCT02,PU89b,DBLP:conf/stoc/AwerbuchBLP89,TZ01}; as another example, spanners with low degree can be used for the design of compact routing schemes~\cite{ABLP90, DBLP:conf/podc/HassinP00, DBLP:conf/podc/AbrahamM04, Tal04,DBLP:conf/podc/Slivkins05,AGGM06,GR08Soda,CGMZ16}. %
Since then, graph spanners have found countless applications in distributed computing as well as various other areas, from motion planning and computational biology to machine learning and VLSI-circuit design.

Spanners have had special success in geometric settings, especially in low-dimensional Euclidean spaces. 
Spanners for Euclidean spaces, namely {\em Euclidean spanners}, were first studied by Chew~\cite{Chew86} in 1986 (even before the term ``spanner'' was coined). Several different constructions of Euclidean spanners enjoy the optimal tradeoff between stretch and size: $(1+\eps)$ versus $O(\eps^{-d}n)$, for $n$-point sets in $\mathbb{R}^d$~\cite{le2019truly}; 
these include $\Theta$-graphs~\cite{Clarkson87, Keil88,KG92,RS91}, Yao graphs~\cite{Yao82}, path-greedy  spanner~\cite{althofer1993,CDNS92,Narasimhan2007}, and the gap-greedy spanner~\cite{Salowe92,AS97}.
The reason Euclidean spanners are so important in practice is that one can achieve stretch arbitrarily close to 1 together with a linear number of edges (ignoring dependencies on $\eps$ and the dimension $d$).
In general metrics, on the other hand, a stretch better than 3 requires $\Omega(n^2)$ edges,
and the best result for general metrics is the same as in general graphs: stretch $2k-1$ with $O(n^{1+1/k})$ edges \cite{PS89,althofer1993}.
Moreover, Euclidean spanners with the optimal stretch-size tradeoff can be built in optimal time $O(n \log  n)$ in the static centralized setting, and they can be distributed in the obvious way in just one communication round in the Congested Clique model.  

Driven by the success of Euclidean spanners, researchers have sought to extend results obtained in Euclidean metrics to the wider family of \emph{doubling metrics}.\footnote{The {\em doubling dimension} of a metric   is the smallest   $d$ s.t.\ every ball of radius $r$ for any $r$ in the metric can be covered by $2^d$ balls of radius $r/2$.
A metric space is called \emph{doubling} if its doubling dimension is constant.} 
The main result in this area is that any $n$-point metric of doubling dimension $d$ admits a $(1+\eps)$-spanner with both sparsity and lightness bounded by $O(\eps^{-O(d)})$~\cite{GGN04,CGMZ16,chan2009small,HM06,Roditty12,GR08Soda,GR082,Smid09,ES15,chan2015new,Sol14,Gottlieb15,BLW17Doubling,filtser2020greedy}.\footnote{In the sequel, for conciseness, we shall sometimes omit the dependencies on $\eps$ and the dimension $d$.} Moreover, here too there are  efficient centralized and distributed algorithms, also under some practical restrictions such as those imposed by Unit Ball Graphs \cite{DPP06a,DPP06b,EFN20,EK21}.

\paragraph{A fundamental drawback of spanners.~}
Different spanner constructions suit different needs and applications.
However, there is one common principle: Once the spanner has been computed, it serves as a ``proxy'' overlay network, on which the computation can proceed, which gives rise to huge savings in a number of quality measures,
including global and local space usage, as well as in various notions of running time, which change from one setting to another; in distributed networks, spanners also lead to additional savings, such as in the message complexity.

Alas, by working on the spanner rather than the original metric, one loses the key property of being able to efficiently ``navigate'' between points. In the  metric, one can go from any point  to any other via a direct edge, which is optimal in terms of  the weighted distance and the unweighted (or {\em hop-}) distance. However, it is unclear how to efficiently navigate in the spanner: How can we translate the {\em existence} of a ``good'' path into an efficient algorithm finding it? 

Moreover, usually by ``good'' path we mean a $t$-spanner path, i.e., a path whose weight approximates the original distance between its endpoints --- but a priori the number of edges (or {\em hops}) in the path could be huge. 
{To control the {\em hop-length} of paths, one can try to upper bound the spanner's hop-diameter, but naturally bounded hop-diameter spanners are more complex than spanners with unbounded hop-diameter, which might render the algorithmic task of efficiently finding good paths more challenging.}
We stress that most existing spanner constructions have inherently high hop-diameters. In particular, any construction with constant degree must have at least a logarithmic hop-diameter, and in general, if the degree is $\Delta$, then the hop-diameter is $\Omega(\log_\Delta n)$.

In Euclidean spaces, the $\Theta$-graph~\cite{Clarkson87, Keil88,KG92,RS91}
and the Yao graph \cite{Yao82} are not only simple spanner constructions, but they also provide {simple} {\em navigation algorithms}, where for any two points $p$ and $q$, one can easily compute a $(1+\eps)$-spanner path between $p$ and $q$. Alas, the resulting path may have a hop-length of $\Omega(n)$, and the {\em query time} is no smaller than the path length. 
There is a $(1+\eps)$-approximate distance oracle for low-stretch spanners \cite{GLNS08}, and while it achieves
constant query time, it does not report the respective paths, whose hop-length can be $\Omega(n)$. 
In doubling metrics, there are $(1+\eps)$-approximate distance oracles with
constant query time \cite{HPM06, BGKLR11}. In \cite{BGKLR11} the respective paths are not part of a sparse overlay network (such as a spanner); in other words, the union of paths returned by the distance oracle of \cite{BGKLR11} may comprise a spanner of $\Theta(n^2)$ edges. Using \cite{HPM06}, one can return paths that are part of a sparse spanner, but their hop-length is $\Theta(\log \rho)$, where $\rho$ -- the metric {\em aspect ratio}, can be arbitrarily large.
This is where bounded hop-diameter spanners may come into play -- efficient constructions are known in Euclidean and doubling metrics \cite{chan2009small,Sol13}. %
In low-dimensional Euclidean spaces, it is possible to build a $(1+\eps)$-spanner with hop-diameter 2 and $O(n \log n)$ edges. In general, for any $k \ge 2$, one can get hop-diameter $k$ with $O(n \alpha_k(n))$ edges, in optimal $O(n \log n)$ time \cite{Sol13}; the function $\alpha_{k}(n)$ is the inverse of a certain function at the $\lfloor k/2 \rfloor$th level of the primitive recursive hierarchy, where 
$\alpha_0(n) = \lceil n/2\rceil$,
$\alpha_1(n) = \lceil \sqrt{n} \rceil$,
$\alpha_2(n)= \lceil\log{n}\rceil$,
$\alpha_3(n)= \lceil\log\log{n}\rceil$,
$\alpha_4(n)= \log^*{n}$,
$\alpha_5(n)= \lfloor \frac{1}{2}\log^*{n} \rfloor$, etc. 
(For $k \ge 4$, the function $\alpha_k$ is close to $\log$ with $\lfloor \frac{k-2}{2} \rfloor$ stars.) 

Two points on the tradeoff curve between hop-diameter $k$ and size $O(n \alpha_k(n))$ deserve special attention: 
(1) $k = 4$ vs.\ $O(n \log^* n)$ edges; in practice $\log^* n \le 10$, i.e., one can achieve hop-diameter 4 with effectively $O(n)$ edges. 
(2) $k = O(\alpha(n))$ vs.\ $O(n \alpha_k(n)) = O(n)$ edges,
where $\alpha$ is a very slowly (more than $\log^*$) growing  function;
so to achieve a truly linear in $n$ edges, one should take a hop-diameter of $O(\alpha(n))$
(which is effectively a constant).
Refer to~\Cref{sec:ackermann} for the formal definitions of the functions $\alpha_k$ and $\alpha$.
In some applications where limiting the hop-distances of paths is crucial, such as in some routing schemes, road and railway networks, and telecommunication, we might need to {\em minimize} the hop-distances; for example, imagine a railway network, where each hop in the route amounts to switching a train -- how many of us would be willing to use more than, say, 4 hops?
Likewise, what if each hop amounts to traversing a traffic light, wouldn't we prefer routes that minimize the number of traffic lights?
In such cases, the designer of the system, or its users, might not be content with super-constant hop-distances, or even with a large constant, and it might be 
of significant value to achieve as small as possible hop-distances.
Motivated by such practical considerations, we are primarily
interested in values of hop-diameter $k$ that ``approach'' 1, mainly $k = 2, 3, 4$, as there is no practical need in considering larger values of $k$ (again, $O(n \alpha_4(n)) = O(n \log^* n)$ edges is effectively $O(n)$ edges).

One can achieve the same result, except for the construction time, also for doubling metrics.
However, as mentioned, the drawback of bounded hop-diameter spanner constructions is that they are far more complex than basic spanners; hence, although there {\em exist} $k$-hop $t$-spanner paths
between all pairs of points, the crux is to {\em find such paths efficiently}.  %

While the original metric  enables us to navigate optimally --- a single hop (for any two points) with the exact distance,
 the price is high --- $\Theta(n^2)$ edges.
The following question naturally arises.

\begin{question} \label{mainq}
Can one efficiently navigate, on a {\em sparse spanner}, using $k$ hops and approximate distances, for $k$ approaching 1? In particular, can we achieve 2, 3 or 4 hops on an $o(n^2)$-sized spanner
in Euclidean or doubling metrics? 
\end{question}

Surprisingly, despite the long line of work on spanners in Euclidean and doubling metrics,~\Cref{mainq} has been overlooked. By ``efficiently navigate'' we mean to {\em quickly} output a path of {\em small weight}, where ideally: (1) ``quickly'' means within time linear in the hop-length of the path, and
(2) ``small weight'' means that the weight of the path would be larger than the original metric distance by at most the stretch factor of the underlying spanner.  

Clearly, \Cref{mainq} can be asked in general, for any class of metrics. To the best of our knowledge, this fundamental question was not asked explicitly before. %
For general graphs, the classic Thorup-Zwick distance oracle \cite{TZ01b} reports $(2\ell-1)$-approximate distance queries in $O(\ell)$ time, using a data structure of expected size $O(\ell n^{1+1/\ell})$; it is immediate that their distance oracle, when applied to metric spaces, can report 2-hop paths of stretch $2\ell-1$ in query time $O(\ell)$, which are all part of the same $(2\ell-1)$-spanner with size $O(\ell n^{1+1/\ell})$. The following question is copied from~\cite{MN06}: 

``{\em Since for large values of distortion (i.e., stretch) the query time of the Thorup-Zwick oracle is large, the problem remained whether
there exist good approximate distance oracles whose query time is a constant independent of the distortion (i.e., in a sense, true "oracles")}''. 
Mendel and Naor~\cite{MN06} gave two distance oracles with $O(1)$ query time and stretch of $128\ell$
(the stretch was improved later to $16\ell$ \cite{NT12}), the first has size $O(n^{1+1/\ell})$ and the respective paths can use any edge of the underlying metric and may thus form a network of size $\Omega(n^2)$, whereas the second has size $O(n^{1+1/\ell} \cdot \ell)$ and the respective paths have hop-lengths $\Theta(\log \rho)$. 
Wulff-Nilsen~\cite{WN13} improved the query time of the Thorup-Zwick distance oracle~\cite{TZ01} to $O(\log{\ell})$. Using the Mendel-Naor distance oracle~\cite{MN06}, Chechik~\cite{Che14,Che15} showed how to improve query time of \cite{TZ01} to $O(1)$,
but this approach suffers from the same drawback --- the respective paths may have hop-lengths $\Theta(\log \rho)$.
Mendel-Naor question can thus be strengthened as:
\begin{question} [Strengthening Mendel-Naor question \cite{MN06}] \label{generalm2}
Is there a good approximate distance oracle for general metric spaces that can report within constant time a constant-hop small-stretch path? 
\end{question}

Interestingly, for planar and minor-free graphs, it is immediate that the respective distance oracles
(\cite{Tho04,KKS11,AG06}), when applied to the respective metrics, can  provide 2-hop paths within constant query time. 

\paragraph{Related work (in a nutshell).}
Thorup~\cite{Tho92} introduced the problem of {\em diameter-reducing shortcuts} for digraphs; the goal is to find a small subset of edges taken from the transitive closure of a digraph so that the resulting digraph has small hop-diameter. 
Cohen \cite{Coh00} introduced the notion of {\em hopsets}; informally, an hopset $H$ is an edge set that, when added to a graph $G$,
provides small-stretch small-hop paths between all vertex pairs.
(See \cite{BLP20,KP22} and references therein for details.) There are also various other related problems, such as \emph{low-congestion shortcuts} \cite{GH16,GH21,KP21}.
For all these problems, the focus is on achieving a graph structure in which there {\em exist} ``good'' paths, i.e., with small hop-length and possibly additional useful properties, between vertex pairs in the graph;
the {\em existence} of such paths found a plethora of applications in distributed, parallel, dynamic and streaming algorithms, such as to the computation of approximate shortest-paths, DFS trees, and graph diameter 
\cite{Ber09,Nan14,MPVX15,HKN18,HL18,GP17,LP19}. 
However, to the best of our knowledge, the computational problem of {\em efficiently reporting} those paths --- which is the focus of our work --- has not been the focus of any prior work.

\subsection{Our contribution}
A key contribution of this work is a conceptual one, in (1) {\em realizing} that it is possible to  efficiently navigate on a much sparser spanner than the entire metric space, and (2) {unveiling} some of applicability of such a navigation scheme. 
We start by considering {\em tree metrics}; a tree metric is a metric for which the distance function is obtained as the shortest-path distance function of some (weighted) tree. 
For any tree metric, when we relax the navigation requirement to use only $k=2$ hops 
(instead of a single hop as in the original metric), we can navigate on a spanner of size $\Theta(n \log {n})$, using 2 hops and stretch 1. %
If we relax the hop-length requirement a bit more, to $k=3$, we can navigate on a yet sparser spanner, of size $\Theta(n\log\log{n})$.
In general, our navigation scheme achieves the same tradeoff between hop-diameter and size as the 1-spanner of Solomon~\cite{Sol13}. 
Our result for navigation on trees is stated in the following theorem (proved in \Cref{subsec:treespanner}); the stretch bound is 1 and one cannot improve the tradeoff between hop-diameter $k$ and size
$\Theta(n\alpha_k(n))$, due to lower bounds by~\cite{Alon87optimalpreprocessing} and~\cite{LMS22} that apply to 1-spanners and (1+$\eps$)-spanners for line metrics, respectively.
\begin{theorem}\label{thm:treeNavigate} Let $M_T$ be any tree metric, represented by an $n$-vertex edge-weighted tree $T$, let $k \ge 2$ be any integer, and let $G_T = (V(T), E)$ be the 1-spanner for $M_T$ with hop-diameter $k$ and $O(n\alpha_k(n))$ edges due to~\cite{Sol13}.
Then we can construct in time $O(n\alpha_k(n))$ a data structure $\mathcal{D}_T$ 
such that, for any two query vertices $u,v\in V(T)$, $\mathcal{D}_T$ returns a 1-spanner path in $G_T$ (which is also a shortest path in $M_T$) between $u$ and $v$ of hop-length $\le k$ in  $O(k)$ time.
\end{theorem}

The runtime of the 1-spanner construction for tree metrics of \cite{Sol13} is $O(n\alpha_k(n))$, hence the data structure provided by~\Cref{thm:treeNavigate} can be built from scratch in time $O(n\alpha_k(n))$.
When it comes to 1-spanners for tree metrics, we can restrict   attention to unweighted trees;
indeed, for any two vertices $u$ and $v$ in tree $T$, if $P_{u, v}$ denotes the unique path between $u$ and $v$ in $T$,  
any 1-spanner path between $u$ and $v$ is a {\em subpath} of $P_{u,v}$ in the underlying tree metric.
\ignore{
It is useful to understand the %
meaning of a 1-spanner path in a tree metric, represented by a (possibly weighted) tree $T$, as in \Cref{thm:treeNavigate}. For any two vertices $u$ and $v$ in $T$, if $P_{u, v}$ denotes the (sequence of vertices on the) unique path between $u$ and $v$ in $T$, then 
any 1-spanner path between $u$ and $v$ is simply a {\em subpath} of $P_{u,v}$ in the underlying tree metric, which thus has exactly the same weight as that of $P_{u,v}$, and as such it is a shortest path between $u$ and $v$ in the   metric.
Note also that this holds regardless of the weight function of the tree $T$, hence, when it comes to 1-spanners for tree metrics, we might as well restrict the attention to unweighted trees in the sequel.
}

Alon and Schieber \cite{Alon87optimalpreprocessing} gave an algorithm for the online tree product that requires $O(n \alpha_k(n))$ time, space and semigroup operations during preprocessing. Their algorithm answers queries following paths of length $2k$, thus achieving $2k$ operations.
This result is equivalent to a linear-time 1-spanner for tree metrics with $O(n \alpha_k(n))$ edges and hop-diameter $2k$, and their query algorithm is in fact a navigation algorithm on top of the underlying 1-spanner.
They also discuss some applications to MST verification, finding maximum flow values in a multiterminal network, and updating a minimum spanning tree after increasing the cost of one of its edges.
\cite{Sol13} presents an improved linear-time construction of 1-spanners for tree metrics, with a hop-diameter of $2k$ rather than $k$ for the same size bound.
Since the 1-spanner construction of \cite{Sol13} is more complex than that of \cite{Alon87optimalpreprocessing}, obtaining a navigation algorithm on top of the 1-spanner of \cite{Sol13} is technically much more intricate than doing so on top of the 1-spanner of \cite{Alon87optimalpreprocessing}.
A central contribution of our work is in obtaining such a navigation algorithm, and then in realizing that, one can extend it to various families of metrics. Moreover, we demonstrate further applicability of our navigation scheme, and also strengthen our results for Euclidean and doubling metrics to achieve fault-tolerance.

To extend the navigation result of~\Cref{thm:treeNavigate} from tree metrics to wider classes of metrics, we apply known results for  \emph{tree covers}, and also design a new {\em robust} tree cover scheme (see \Cref{lm:robust-treecover}). 
Let $M_X = (X, \delta_X)$ be an arbitrary metric space.
We say that a weighted tree $T$ is a {\em dominating} tree for $M_X$ if $X \subseteq V(T)$ 
and it holds that $\delta_T(x,y) \ge \delta_X(x,y)$, for every $x,y\in X$.
For $\gamma \ge 1$ and an integer $\zeta\ge1$, a {\em $(\gamma, \zeta)$-tree cover} of   $M_X = (X,\delta_X)$ is a collection of $\zeta$ {\em dominating trees} for $M_X$, such that for every $x,y \in X$, there exists a tree $T$ with $d_T(u,v) \le \gamma \cdot \delta_X(u,v)$;
we say that the {\em stretch} between $x$ and $y$ in $T$ is at most $\gamma$,
and the parameter $\gamma$ is referred to as the {\em stretch} of the tree cover. 
A tree cover is called a \emph{Ramsey tree cover}
if for each $x \in X$, there exists a ``home'' tree $T_x$, such that the stretch 
between $x$ and every other vertex $y \in X$ in $T_x$ is at most $\gamma$.

The celebrated ``Dumbbell Theorem'' by Arya et al.~\cite{ADMSS95} provides a $(1+\eps,O(\frac{\log(1/\eps)}{\eps^d}))$-tree cover %
 in $O(\frac{\log(1/\eps)}{\eps^d}\cdot n\log{n}+\frac{1}{\eps^{2d}}\cdot n)$ time,
for $d$-dimensional Euclidean spaces.
For general metrics, the seminal work
of Mendel-Naor \cite{MN06}
provides a {\em Ramsey $(\gamma, \zeta)$-tree cover} with $O(\zeta n^{2+1/\zeta} \log n)$ time, where $\gamma = O(\ell), \zeta = 
O(\ell \cdot n^{1/\ell})$ for any $\ell \ge 1$.
Additional tree cover constructions are given in \cite{BFN19B}, including
a $(1+\eps, (1/\eps)^{\Theta(d)})$-tree cover for metrics with doubling dimension $d$.
(See \Cref{tab:treeCover} in \Cref{app:tree-covers}.)
Plugging \Cref{thm:treeNavigate} on these tree cover theorems, we obtain: 

\begin{theorem}\label{thm:oracle}
For any $n$-point metric $M_X = (X, \delta_X)$
and any integer $k \ge 2$, one can construct  a $\gamma$-spanner $H_X$ 
for $M_X$ with hop-diameter $k$ and $O(n \alpha_k(n) \cdot \zeta)$ edges,
accompanied with a data structure $\mathcal{D}_{X}$, such that for any two query points $u,v\in X$, $\mathcal{D}_X$ returns in time $\tau$ a $\gamma$-spanner path in $H_X$ between $u$ and $v$ of at most $k$ hops, where 
	\begin{itemize} [itemsep=0.4pt,topsep=0.4pt,parsep=0.4pt]
		\item $\gamma=(1+\eps)$, $\zeta=(1/\eps)^{\Theta(d)}$, $\tau=O(k/\eps^{\Theta(d)})$, if the doubling dimension of $M_X$ is $d$. 
		\item If $M_X$ is a general metric, there are two possible tradeoffs, for any integer $\ell \ge 1$:
		\begin{itemize}%
		\item
		$\gamma = O(\ell)$, $\zeta = O(\ell \cdot n^{1/\ell})$, $\tau=O(k)$.
		\item 
		$\gamma = O(n^{1/\ell}\cdot \log^{1-1/\ell}{n})$, $\zeta = \ell$, $\tau=O(k)$.
		\end{itemize}
			\item $\gamma = (1+\eps)$, $\zeta = O(((\log{n})/\eps)^2)$, $\tau=O(k\cdot ((\log{n})/\eps)^2)$, if $M_X$ is a fixed-minor-free metric.
	\end{itemize}
	If $M_X$ is doubling,   the running time is $O(n\log{n})$, for fixed $\eps$ and constant dimension $d$.
\end{theorem}

The navigation algorithms provided by \Cref{thm:oracle} 
work by first determining the right tree for the query points $u,v \in X$,
and then applying the tree navigation algorithm of \Cref{thm:treeNavigate} on that tree. This two-step navigation scheme might be advantageous over  navigation algorithms that don't employ trees, as navigation on top of a tree could be both faster and simpler to implement in practice.
 \Cref{thm:oracle}  implies that in low-dimensional and doubling metrics, one can navigate along a $(1+\eps)$-spanner with hop-diameter $k$ and $O(n \alpha_k(n))$ edges, within query time $O(k)$, ignoring dependencies on $\eps$ and $d$. 
Result of this sort was not known before even in Euclidean spaces, and it affirmatively settles \Cref{mainq}.
In metrics induced by fixed-minor-free graphs (e.g., planar metrics), we get a similar result,
with the number of edges and query time growing by a factor of $\log^2 n$. 
For such metrics, as mentioned, 
there are already efficient navigation algorithms, implicit in \cite{Tho04,KKS11,AG06}, so we do not achieve improved bounds here; however, as argued above, our two-step navigation scheme might still be advantageous.
Finally, in general metrics, the stretch and size of the spanners on which we navigate nearly match the best possible stretch-size tradeoff of spanners in general metrics, and the number of hops in the returned paths approaches 1.
Here too, there are already efficient navigation algorithms, which achieve better bounds on stretch and size, implicit in the works of 
\cite{TZ01,MN06,NT12,Che14,Che15}. However, our two-step navigation scheme in general metrics is advantageous over previous ones since 
it reports an actual path that belongs to the underlying spanner {\em in constant time}, which also settles \Cref{generalm2};
moreover, it uses a Ramsey cover, and is thus of further applicability (e.g.,  for routing protocols, see below). %

\ignore{
\shay{But there are two points that we should emphasize (Lazar, please remember to include this somewhere in the paper, and color it red):
1) The net-tree spanner is not as basic as a tree (due to cross edges), and there's an advantage in a tree vs. a net-tree spanner (even though it's a collection of trees and not just one). E.g., going from a node to its parent in the tree can be implemented quickly/easily, and roughly the same can be said about LCA, whereas determining whether two nodes share a cross edge can be viewed as more complex and slow (even though the deg of nodes (rather than points) in the net-tree is constant).
For example, if we want to "navigate" between two leaves u and v in the tree, we simply go up the tree until they meet. This is also relevant in routing.
2) For our purposes of achieving a very small hop-diam, say 2 or 3 - this is impossible to achieve by shortcutting the net-tree spanner - there the diameter will not be k but rather 2k+1. So e.g. instead of 2 hops, we'll get 5 hops.}
}

\paragraph{A unified approach.}
Although our original motivation was in Euclidean spaces, 
our two-step navigation scheme extends far beyond it. %
Our technique for efficiently navigating 1-spanners for tree metrics, as provided by \Cref{thm:treeNavigate}, 
provides a {\em unified reduction} from efficient navigation schemes in an arbitrary metric class to {\em any tree cover} theorem in that class; in other words, {any new tree cover theorem  will directly translate into a new navigation scheme}.

\paragraph{A fault-tolerant spanner and navigation scheme.}
In Euclidean and doubling metrics, we design a fault-tolerant (FT) navigation scheme, 
where we can navigate between pairs of non-faulty points in the network even when a predetermined number $f$ of nodes become faulty,
while incurring small overheads (factor of at most $f$) on the size of the navigation data structure and other parameters. 
We first  generalize the Euclidean ``Dumbbell Tree'' Theorem ~\cite{ADMSS95} for  doubling metrics; this generalization is nontrivial and is perhaps the strongest technical contribution of this work. 
At a high-level, the ``Dumbbell Tree'' Theorem is quite {\em robust} against adversarial perturbations of input points;
specifically, any internal node in any tree in the cover can be assigned any descendant leaf as its associated point without affecting the stretch bound. This property  %
is not achieved by the tree cover of \cite{BFN19B} in doubling metrics.
Building on our robust tree cover theorem, we design a new construction of FT sparse spanners of bounded hop-diameter; this construction achieves optimal bounds on all involved parameters for fixed $f$, and is of independent interest.
Our FT navigation scheme is obtained from our new FT spanner
just as our basic navigation scheme is obtained from the basic spanner of \cite{Sol13}. 
See \Cref{ft} for the full details.

\paragraph{Broad applicability.}
We argue that an efficient navigation scheme is of broad potential applicability,
by providing a few applications and implications;
we anticipate that more will follow.

Perhaps the main application of our navigation technique is an efficient {\em routing scheme}, where we achieve small bounds on the  local memory at all nodes, even though the maximum degree is huge, which is inevitable for spanners of tiny hop-diameter.
Due to space constraints, in this  discussion we provide details only on this application. 
In a nutshell, other applications of our navigation scheme include: (1) {\em Efficient sparsification of light-weight spanners}, where we start from an arbitrary light-weight but possibly dense spanner and transform it into a spanner that has the original stretch and weight but is also sparse. (2) {\em Efficient computation on the spanner}, where we are able to compute basic graph structures (such as MST and SPT) efficiently on top of a spanner rather than the underlying metric (which is not as part of our input).
(3) {\em Online tree product queries and applications}, where our basic navigation scheme can be used as a query algorithm 
for the online tree product problem, which finds applications to MST verification and other problems.
More details on these applications are deferred to  \Cref{appNutshell} (introductory details) and  
\Cref{sec:applications} (full details). 

Our basic result on routing schemes is in providing a routing scheme of stretch 1 on \emph{tree metrics}, for $k=2$ hops and using labels and local routing tables of $O(\log^2{n})$ bits and headers of $O(\log{n})$ bits. The routing scheme works in the labeled, fixed-port model (see \Cref{sec:labeling} for the definitions). The bound on the number of hops is best possible without routing on the complete graph. 
We employ this basic routing scheme in conjunction with the aforementioned tree covers and obtain efficient routing schemes for 
doubling, general and fixed-minor-free metrics. %
For doubling metrics, we strengthen the result to achieve
a fault-tolerant routing scheme, where packets can be routed efficiently even when a predetermined number of nodes in the input metric become faulty. 

\begin{theorem}\label{thm:routing}
For any $n$-point metric $M_X = (X, \delta_X)$, one can construct a $\gamma$-stretch $2$-hop routing scheme in the labeled, fixed-port model with headers of $\lceil\log{n}\rceil$ bits, labels of $b_l$ bits, local routing tables of $b_t$ bits, and local decision time $\tau$, where:
\begin{itemize} [itemsep=0.4pt,topsep=0.4pt,parsep=0.4pt]
\item $\gamma=(1+\eps)$, $b_l=b_t=O(\eps^{-O(d)} \log(n)\log(n/\eps))$, $\tau = O(\eps^{-O(d)}),$ for doubling dimension $d$.
\item If $M_X$ is a general metric, there are two possible tradeoffs, for any integer $\ell\ge1$:
\begin{itemize}%
\item $\gamma=O(\ell)$, $b_l=O(\log^2{n})$, $b_t=O(\ell \cdot n^{1/\ell} \log^2{n})$, $\tau=O(1)$.
\item $\gamma=O(n^{1/\ell}\cdot \log^{1-1/\ell}{n})$, $b_l=O(\log^2{n})$, $b_t=O(\ell \log^2{n})$, $\tau=O(1)$.
\end{itemize}
\item $\gamma=(1+\eps)$, $b_l=b_t=O((\log{n}/\eps)^3 \log{n})$, $\tau=O((\log{n}/\eps)^2)$, for a fixed-minor-free metric.
\end{itemize}
If $M_X$ has doubling dimension $d$, the running time is $O(n\log{n})$, for fixed $\eps$ and $d$. 
In this case, one can achieve an $f$-fault-tolerant routing scheme, with the bounds on $b_l$ and $b_t$ growing by a factor of $f$.
\end{theorem}

This provides the first routing schemes in Euclidean as well as doubling metrics, 
where the number of hops is as small as 2, and the labels have near-optimal size. 
To the best of our knowledge, no previous work on routing schemes in Euclidean or doubling metrics achieve a sub-logarithmic bound on the hop-distances, let alone a bound of 2. Some previous works \cite{GR08Soda, CGMZ16}
obtain their routing schemes by routing on constant-degree spanners, which means that the hop-diameters of those spanners are at least $\Omega(\log n)$, hence the hop-lengths of the routing paths are $\Omega(\log n)$ too. The other routing schemes \cite{DBLP:conf/podc/HassinP00, DBLP:conf/podc/AbrahamM04, Tal04, DBLP:conf/podc/Slivkins05, AGGM06} do not work in this way, but still have a hop-diameter of $\Omega(\log{n})$ or even $\Omega(\log \rho)$. 
We also stress that our routing scheme is fault-tolerant, which is of practical importance, and we are not aware of any previous fault-tolerant routing scheme in Euclidean or doubling metrics.

There are many works on routing in general graphs \cite{ABLP90,AP92,Cowen01,TZ01,EGP03,Che13,RT15,ACEFN20,Fil22}. In metrics, it is much easier to get an efficient routing scheme. The Thorup-Zwick routing scheme \cite{TZ01} can achieve two hops in general metrics with stretch $4\ell-5$ (improved to $3.68\ell$ \cite{Che13}), labels of $O(\ell \log n)$ bits, and table sizes of $\tilde{O}(n^{1/\ell})$. 
These approaches, when modified to work in metrics, incur a decision time of $O(\ell)$, and it is not clear whether it can be improved. Our result for general metrics from \Cref{thm:routing},
while inferior in terms of the stretch (a constant factor), the label sizes (a $\log n / \ell$ factor) and table size (a $\log n \ell$ factor), achieve constant decision time, which might be an important advantage in real-time routing applications. 

\subsection{Further discussion on Applications} \label{appNutshell}

\paragraph{Efficient sparsification of light spanners.}
Let $M_X = (X,\delta_X)$ be an arbitrary $n$-point metric space and let $G$ be any $m$-edge spanner for $M_X$ of light weight.  %
Our goal is to transform $G$ into a sparse spanner for $M_X$, without increasing the stretch and weight by much. 
Let $\mathcal{D}_X$ be the data structure provided by~\Cref{thm:oracle}.
For each edge in $G$, we can query $D_X$ for the $k$-hop path between its endpoints and then return the union of the paths over all edges. It is not difficult to verify that the resulting graph is a spanner for $M_X$, whose stretch and weight are larger than those of $G$ by at most a factor $\gamma$, but it includes at most $O(n\alpha_k(n)\cdot \zeta)$ edges --- thus it is not only light but also sparse. The runtime of this transformation is $O(m \cdot \tau)$. (As in \Cref{thm:oracle}, we denote by $\gamma$ the stretch of the tree cover, $\zeta$ bounds the number of trees in the cover, and $\tau$ bounds the query time --- which is $O(k \log^2 n)$ for fixed-minor-free graphs, 
and $O(k)$ for all other metric classes.) For further details, see \Cref{sec:sparse-spanner}.

\paragraph{Efficient computation on the spanner.}
As mentioned already, once a spanner has been constructed, it usually serves as a ``proxy'' overlay network, on which any subsequent computation can proceed,
in order to obtain savings in various measures of space and running time.
This means that any algorithm that we may wish to run, should be (ideally) run on top of the spanner itself. Furthermore, in some applications, we may not have direct access to the entire spanner, but may rather have implicit and/or local access, such as via labeling or routing schemes, or by means of a data structure for approximate shortest paths within the spanner, such as the one provided by~\Cref{thm:oracle}.

Suppose first that we would like to construct a (possibly approximate) shortest path tree (SPT).
An SPT for the original metric space is simply a star (in any metric). 
But the star is (most likely) not a subgraph of the underlying spanner. How can we efficiently transform the star into an approximate SPT in the spanner?
If we have direct, explicit access to the spanner, we can simply compute an SPT on top of it using Dijkstra's algorithm, which will provide an approximate SPT for the original metric. 
Dijkstra's algorithm, however, will require $\Omega(n \log n)$ time (for an $n$-vertex spanner), even if the spanner size is $o(n \log n)$; there is also another SPT algorithm that would run in time linear in the spanner size, but it is more complex and also assumes that $\log n$-bit integers can be multiplied in constant time \cite{DBLP:journals/jacm/Thorup99}. 
Using our navigation scheme, as provided by~\Cref{thm:oracle}, we can do both better and simpler, and we don't even need explicit access to the underlying spanner (though we do need, of course, access to the navigation scheme).  %
The data structure provided by~\Cref{thm:oracle} allows us to construct, within time $O(n \tau)$, an
approximate SPT. In particular, for low-dimensional Euclidean and doubling metrics, 
we can construct a $(1+\eps)$-approximate SPT (for a fixed $\eps$) that is a subgraph of the underlying spanner within $O(n k)$ time, where $k = 2,3,\ldots,O(\alpha(n))$. Refer to \Cref{sec:spt} for further details.

Suppose next that we would like to construct an approximate minimum spanning tree (MST). 
In low-dimensional Euclidean spaces one can compute a $(1+\eps)$-approximate MST (for a fixed $\eps$) in $O(n)$ time \cite{chan2008well}, 
but again this approximate MST may not be a subgraph of the spanner. Running an MST algorithm on top of the spanner would require time that is at least linear in the spanner size; moreover, the state-of-the-art deterministic algorithm runs in super-linear time and is rather complex \cite{Chazelle00},
and the state-of-the-art linear time algorithms either rely on randomization \cite{KKT95} or on some assumptions, such as the one given by \emph{transdichotomous model} \cite{FW94}. Instead, using our navigation scheme, as provided by~\Cref{thm:oracle}, we can construct an approximate MST easily, within time $O(n \tau)$. In particular, for low-dimensional Euclidean spaces, we can construct in this way a $(1+\eps)$-approximate MST (for a fixed $\eps$) that is a subgraph of the underlying spanner within $O(n k)$ time, where $k = 2,3,\ldots,O(\alpha(n))$. Refer to \Cref{sec:mst} for further details.

The same principle extends to other metric spaces, but some of the guarantees degrade.
In particular, for metric spaces, the approximation factor increases far beyond $1+\eps$, at least assuming we would like the size of the underlying spanner to be near-linear. 
We stress that (approximate) SPTs and MSTs are two representative examples, but the same approach can be used for efficiently constructing other subgraphs of the underlying spanner; we also note that a shallow-light tree (SLT) \cite{ABP90, ABP92, KRY93,Sol14}, which is tree structure that combines the useful properties of an SPT and an MST, can be constructed in linear time given any approximate SPT and MST, and the resulting SLT is also a subgraph of these input trees \cite{KRY93}. Thus, after constructing approximate SPT and MST as described above, we obtain, within an additional linear time, an SLT that is a subgraph of the underlying spanner.

\paragraph{Online tree product and MST verification.}
The paper by \cite{Alon87optimalpreprocessing} focuses on the following problem. Let $T$ be a tree with each of its $n$ vertices being associated with an element of a semigroup. One needs to answer online queries of the following form: Given a pair of vertices $u,v\in T$, find the product of the elements associated with the vertices along the path from $u$ to $v$.
They show that one can preprocess the tree using $O(n\alpha_k(n))$ time and space, so that each query can be answered using at most $2k$ semigroup operations. They also showed several applications of their algorithm, such as to finding maximum flow in a multiterminal network, MST verification, and updating the MST after increasing the cost of its edges \cite{Alon87optimalpreprocessing}.

One can show that the result of \cite{Alon87optimalpreprocessing} gives rise to a construction of 1-spanners for tree metrics with $O(n\alpha_k(n))$ edges and hop-diameter $2k$. This is inferior to the spanner of \cite{Sol13}, since it achieves a twice larger hop-diameter ($2k$ instead of $k$) for the same number of edges, $O(n\alpha_k(n))$. In particular, their construction cannot be used to achieve hop-diameters 2 and 3, which is the focus of this paper.
We stress that some of the applications that we discussed above cannot be achieved using this weaker result. As a prime example, our routing scheme crucially relies on having hop-diameter 2. Paths of hop-distance 2 have a very basic structure (going through a single intermediate node), which our routing scheme exploits. As a result, the underlying spanner has $O(n \log n)$ edges,
which ultimately requires us to use $ O(\log^2 n)$ bits of space. Whether or not one can use a spanner of larger (sublogarithmic and preferably constant) hop-diameter for designing compact routing schemes with $o(\log^2 n)$ bits is left here as an intriguing open question. Exactly the same obstacle should render the construction of \cite{Alon87optimalpreprocessing} infeasible for constructing efficient routing schemes, since the hop-distances provided by the result of \cite{Alon87optimalpreprocessing} are larger than 2.

In \Cref{sec:treeProd}, we show that our navigation scheme (provided by \Cref{thm:treeNavigate}) can be used as a query algorithm that supports all the applications supported by \cite{Alon87optimalpreprocessing}, but within a factor 2 improvement on the hop-distances (or on other quality measures that are derived from the hop-distances). One such application is to the online MST verification problem, which is the main building block for randomized MST algorithms.
For this problem, Pettie \cite{DBLP:journals/combinatorica/Pettie06} shows that it suffices to spend $O(n  \alpha_{2k}(n))$ time and space and $O(n \log \alpha_{2k}(n))$ comparisons during preprocessing, so that each subsequent query can be answered using $4k-1$ comparisons.\footnote{In fact, Pettie \cite{DBLP:journals/combinatorica/Pettie06} claimed that the preprocessing time is $O(n\log\alpha_{2k}(n))$, and that each subsequent query can be answered using $2k-1$ comparisons. This is inaccurate, as we elaborate in \Cref{sec:mstVerify}.}
Our algorithm takes $O(n  \alpha_{2k}(n))$ time and space and $O(n \log \alpha_{2k}(n))$ comparisons during preprocessing, so that each subsequent query is answered using $2k-1$ comparisons in $O(k)$ time. The result of \cite{DBLP:journals/combinatorica/Pettie06} can also achieve a query time of $O(k)$, by building on \cite{Alon87optimalpreprocessing}, but using $4k-1$ comparisons rather than $2k-1$ as in our result.
Concurrently and independently of us, Yang \cite{DBLP:journals/corr/abs-2105-01864} obtained a similar result.

\section{Preliminaries}\label{sec:prelim}

This section contains definitions and results required for the rest of the paper. In particular, \Cref{app:tree-covers} summarizes known tree cover theorems which we rely on and \Cref{sec:ackermann} introduces variants of Ackermann function which we use.

\subsection{Summary of known results on tree covers}\label{app:tree-covers}
The following table summarizes known results on tree cover theorems. 

\begin{table}[H]
\hspace*{-1.2cm}
\centering
\scalebox{0.90}{
\begin{tabular}{ |c | c | c |c|c| }
\hline
\textbf{Stretch ($\bm{\gamma}$)} & \textbf{Num. of trees ($\bm{\zeta}$)} & \textbf{Metric family} & \textbf{Construction time} & \textbf{Authors} \\\hline
$1+\eps$ & $(1/\eps)^{\Theta(d)}$ &  with doubling dim. $d$ & $O(n\log{n})$ &  \cite{ADMSS95, BFN19B}\\\hline
$1+\eps$ & $O(((\log{n})/\eps)^2)$ & fixed-minor-free (e.g., planar) & $n^{O(1)}$  &  \cite{BFN19B}\\\hline
$O(\ell)$ & $O(\ell \cdot n^{1/\ell})$  & general  &  $O(\ell \cdot n^{2+1/\ell}\log{n})$&  \cite{MN06} \\\hline
$O(n^{1/\ell}\cdot \log^{1-1/\ell}{n})$ & $\ell$ & general  &  $n^{O(1)}$ & \cite{BFN19B} \\\hline
\end{tabular}
}
\caption{Summary of the tree cover results used throughout the paper. The last two results provide Ramsey tree covers for any integer $\ell\ge 1$.}\label{tab:treeCover}
\end{table}

\subsection{Ackermann functions}\label{sec:ackermann}
Following standard notions \cite{DBLP:journals/jacm/Tarjan75,Alon87optimalpreprocessing,Chaz87,Narasimhan2007,Sol13}, we will introduce two very rapidly growing functions $A(k, n)$ and $B(k, n)$, which are variants of Ackermann's function. Later, we also introduce several inverses and state their properties that will be used throughout the paper.
\begin{definition}
For all $k \ge 0$, the functions $A(k,n)$ and $B(k,n)$ are defined as follows:
\begin{align*}
A(0, n) &\coloneqq 2n, \text{\emph{ for all }} n \ge 0,\\
A(k, n) &\coloneqq 
\begin{cases}
1 &\text{\emph{ if }} k \ge 1 \text{\emph{ and }} n = 0\\
A(k-1, A(k, n-1)) & \text{\emph{ if }} k \ge 1 \text{\emph{ and }} n \ge 1\\
\end{cases}\\
B(0, n) &\coloneqq n^2, \text{\emph{ for all }} n \ge 0,\\
B(k, n) &\coloneqq 
\begin{cases}
2 &\text{\emph{ if }} k \ge 1 \text{\emph{ and }} n = 0\\
B(k-1, B(k, n-1)) & \text{\emph{ if }} k \ge 1 \text{\emph{ and }} n \ge 1\\
\end{cases}
\end{align*}
\end{definition}
We now define the functional inverses of $A(k,n)$ and $B(k,n)$.
\begin{definition}%
For all $k \ge 0$, the function $\alpha_k(n)$ is defined as follows:
\begin{align*}
 \alpha_{2k}(n) &\coloneqq \min\{s \ge 0: A(k, s) \ge n\},\text{\emph{ for all }} n\ge 0\text{,}\\
	 \alpha_{2k+1}(n) &\coloneqq \min\{s \ge 0: B(k, s) \ge n\},\text{\emph{ for all }} n\ge 0.
\end{align*}
\end{definition}
It is not hard to verify that $\alpha_0(n) = \lceil n/2\rceil$,
$\alpha_1(n) = \lceil \sqrt{n} \rceil$,
$\alpha_2(n)= \lceil\log{n}\rceil$,
$\alpha_3(n)= \lceil\log\log{n}\rceil$,
$\alpha_4(n)= \log^*{n}$,
$\alpha_5(n)= \lfloor \frac{1}{2}\log^*{n} \rfloor$, etc. 

The spanner construction of \cite{Sol13} (and thus also our navigation algorithm) uses a slight variant $\alpha'_k$ of the function $\alpha_k$.
\begin{definition}
We define function $\alpha_k'(n)$ as follows:
\begin{align*}
\alpha'_0(n) &\coloneqq \alpha_0(n), \text{\emph{ for all }} n \ge 0,\\
\alpha'_1(n) &\coloneqq \alpha_1(n), \text{\emph{ for all }} n \ge 0,\\
\alpha'_k(n) &\coloneqq \alpha_k(n), \text{\emph{ for all }} k\ge 2 \text{\emph{ and }} n \le k+1,\\
\alpha'_k(n) &\coloneqq 2 + \alpha'_k(\alpha'_{k-2}(n)), \text{\emph{ for all }} k\ge 2 \text{\emph{ and }} n\ge k+2.\\
\end{align*}
\end{definition}
By Lemma 2.4 from \cite{Sol13} we know that the functions $\alpha'_k$ and $\alpha_k$ are asymptotically close --- for all $k,n \ge 0$, $\alpha_k(n) \le \alpha_k'(n) \le 2\alpha_k(n)+4$.

Finally, for all $n\ge 0$, we introduce the Ackermann function as $A(n) \coloneqq A(n,n)$, and its inverse as $\alpha(n) = \min\{s \ge 0 : A(s) \ge n\}$. In \cite{Narasimhan2007}, it was shown that  $\alpha_{2\alpha(n)+2}(n)\le 4$.

Following \cite{DBLP:journals/combinatorica/Pettie06}, we introduce a slight variant of Ackermann's function as follows:\footnote{In \cite{DBLP:journals/combinatorica/Pettie06}, the function was denoted by letter $A$.}
\begin{align*}
P(1,j) &\coloneqq 2^j & \text{ for } j\ge 0,\\
P(i, 0)  &\coloneqq P(i-1, 1) & \text{ for } i \ge 2,\\
P(i, j) &\coloneqq P(i-1,2^{2^{P(i, j-1)}} ) & \text{ for } i \ge 2, j \ge 1.
\end{align*}
Then, its inverse of the $i$th row is defined as:
\begin{align*}
\lambda_i(n) \coloneqq \min\{j \ge 0 : P(i, j)\ge n\}.
\end{align*}

\begin{restatable}{lemma}{compareAckermann}\label{lemma:compareAckermann}
For any $i\ge 1$, if $\lambda_i(n) >0$, then $\frac{1}{3}  \alpha_{2i}(n) \le \lambda_i(n) \le \alpha_{2i}(n)$. 
\end{restatable}

\section{Navigating metric spaces}\label{SectionPathOracle}
In this section we present the navigation algorithm for metric spaces.
\Cref{subsec:treespanner} is devoted to proving \Cref{thm:treeNavigate}, which concerns navigation on 1-spanner with bounded hop-diameter for tree metrics by \cite{Sol13}. In \Cref{subsec:treeCoverNav}, we use tree cover theorems (cf.~\Cref{tab:treeCover}) and prove \Cref{thm:oracle}, which concerns navigation on metric spaces.

\subsection{Navigating the tree spanner}\label{subsec:treespanner}

Our navigation algorithm consists of two parts. In \Cref{subsec:preprocess-tree}, we present a preprocessing algorithm, which takes a tree $T$ and an integer parameter $k \ge 2$; it constructs \cite{Sol13} 1-spanner $G_T$ with hop-diameter $k$ for a tree metric $M_T = (V(T), \delta_T)$ induced by $T$, together with data necessary for efficiently navigating it. Next, in \Cref{subsec:query}, we present a query algorithm which, given any two vertices $u,v \in T$, outputs in $O(k)$ time a 1-spanner $k$-hop path between $u$ and $v$ in $G_T$.

The result of \cite{Sol13} considers a generalized problem of constructing 1-spanners for \emph{Steiner tree metrics}. Specifically, suppose that in a given tree $T$, a subset $R(T) \subseteq V(T)$ of the vertices are set as \emph{required vertices}. The other vertices $S(T) \coloneqq V(T)\setminus R(T)$ are called \emph{Steiner vertices}. We say that a 1-spanner $G_T$ for $M_T$ has hop-diameter $k$ if it contains a 1-spanner path for $M_T$ that consists of at most $k$ edges, for every pair of vertices in $R(T)$. 

We next give high-level explanation of the \cite{Sol13} spanner construction algorithm. It relies on the following two procedures, which we shall also use in our preprocessing algorithm.
\begin{itemize}
\item $\Call{Prune}{(T, rt(T)), R(T)}$: Takes as an input a tree $T$, its root $rt(T)$, and the set of required vertices $R(T)$. Outputs an edge-weighted tree $(T_{\pruned}, rt(T_\pruned))$, which contains $R(T)$ and has at most $|R(T)|-1$ Steiner vertices. We set the weight $w_{T_{\pruned}}(u,v) = \delta_T(u,v)$ for every edge $(u,v)\in V(T_{\pruned})$. Informally, the procedure keeps the intrinsic properties of $T$, while reducing the number of Steiner vertices. For more details, see Section 3.2 in \cite{Sol13}. The running time is $O(|V(T)|)$.
\item $\Call{Decompose}{(T, rt(T)), R(T), \ell}$: Takes as an input a rooted tree $(T,rt(T))$, the set of required vertices $R(T)$, and an integer parameter $\ell \geq 1$. (The parameter $\ell$ will be set to $\alpha'_{k-2}(n)$.) Outputs, in $O(|V(T)|)$ time, a set of \emph{cut vertices}, denoted by $CV_{\ell} \subseteq V(T)$, such that every connected component (tree) of $T\setminus CV_{\ell}$ contains at most $\ell$ required vertices. The size of $CV_{\ell}$ satisfies $|CV_{\ell}| \leq \lfloor \frac{|V(T)|}{\ell+1} \rfloor$. 
\end{itemize}

At the beginning of the spanner construction, we find a subset of vertices $CV_{\ell} \subseteq V(T)$, using $\Call{Decompose}{(T, rt(T)), R(T), \ell}$. 
We then compute the set of edges $E'$, which interconnects vertices in $CV_\ell$. The algorithm distinguishes several cases:
\begin{itemize}
\item If $k = 2$, then $|CV_{\ell}| = 1$ and $E' = \emptyset$.
\item If $k = 3$, then connect every pair of vertices in $CV_\ell$, i.e., $E' = CV_\ell \times CV_\ell$.
\item If $k\geq 4$, then make a copy $T'$ of $T$, set $CV_\ell$ as its required vertices and prune it, by invoking $\Call{Prune}{(T', rt(T')), CV_\ell}$; let $E'$ be the set of edges returned by recursive spanner construction on $T'$ with hop-diameter set to $k-2$.
\end{itemize}
Denote by $T_1, \ldots, T_p$ the trees in $T \setminus CV_\ell$.
The algorithm computes the set of edges $E''$ that connects the cut vertices of $CV_\ell$ with the corresponding subtrees.
Given a subtree $T_i\in T$, we say that a vertex $u\in T$ is a \emph{border vertex} of $T_i$ if $u\notin V(T_i)$ is adjacent to a vertex in $T_i$. Let $\border(T_i)$ denote the set of all border vertices of $T_i$. With a slight abuse of notation, we let $\border(v) = \border(T_i)$ for all $v \in T_i$; in addition, for $c \in CV_\ell$, let $\border(c) = \{v \in T \mid c \in \border(v)\}$.  
For every $c \in CV_\ell$, we add an edge from $c$ to all the required vertices in $\border(c)$. Finally, for each $i$ in $[p]$, we let $E_i$ be the set of edges obtained by recursive spanner construction on $T_i$. The set of spanner edges is $E' \cup E'' \cup \bigcup_{i \in [p]} E_i$. This concludes the high-level description of algorithm for constructing spanner.

The construction guarantees that between any two vertices $u,v\in R(T)$, there is a path of length $\delta_T(u,v)$ in $G_T$ consisting of at most $k$ edges. This path is a shortcut of the path between $u$ and $v$ in $T$. More formally, denote by $\mathcal{P}_T(u,v)$ the unique path in $T$ between a pair $u,v$ of vertices in $T$. A path $P$ in $G_T$ between $u$ and $v$ is called $T$-monotone if it is a subpath of $\mathcal{P}_T(u,v)$, that is, if $\mathcal{P}_T(u,v) = (u = v_0,v_1,\ldots,v_t = v)$, then $P$ can be written as $P = (u = v_{i_0},v_{i_1},\ldots,v_{i_q} = v)$, where $0 = i_0 < i_1 < \dots < i_q = t$. For any two vertices $u,v \in R(T)$, there is a $T$-monotone path in $G_T$ of at most $k$ edges.

Despite the guarantee of existence of a $k$-hop path between any two vertices in $R(T)$, it is not a priori clear how one can efficiently find such a path. Consider $k=2$, as the most basic setting. 
It is shown in \cite{Sol13} that for any two $u$ and $v$, there exists an intermediate cut vertex $w$ on the path $\mathcal{P}_T(u,v)$, such that $(u, w)$ and $(w, v)$ are in $G_T$. (For simplicity, we omit some technical details of handling the corner cases.) But this cut vertex can be anywhere on $\mathcal{P}_T(u,v)$ and (at least naively) finding it could take number of steps linear in the length of the path.

Our key idea is to rely on the recursion tree of the spanner construction algorithm. Since the edges $(u,w)$ and $(w,v)$ are in $G_T$ and $w$ is a cut vertex, there must be a recursive call which had $CV_\ell = \{w\}$.
We explicitly build the recursion tree of the spanner construction, and store with each of its vertices the data required for efficient navigation. We call such a tree \emph{augmented recursion tree},  and denote it by $\Phi$. For each vertex $v$ in $R(T)$, we keep track of the vertex in $\Phi$ which corresponds to the recursive call when $v$ was chosen to be a cut vertex. To answer a query for $k$-hop path between $u$ and $v$, we can find an intermediate cut vertex $w$ as follows. First, we identify two vertices $\alpha_u$ and $\alpha_v$ corresponding to $u$ and $v$ in $\Phi$. Then, we find their lowest common ancestor $\beta$ in $\Phi$. Vertex $\beta$ corresponds to a recursive call in which some cut vertex $w$ splits the tree so that $u$ and $v$ are in different subtrees. Clearly, $w$ is on $\mathcal{P}_T(u,v)$. Since $u$ and $v$ are both required vertices and $w$ is a cut vertex, the edges $(w,u)$ and $(w,v)$ are added to the spanner in this recursive call. Hence, we have found a $T$-monotone 2-hop path between $u$ and $v$ in $G_T$. 

When $k=3$, the set of cut vertices at each recursion level contains more than one cut vertex. The 1-spanner path between $u$ and $v$ contains two intermediate cut vertices, say $u'$ and $v'$, which are on $\mathcal{P}_T(u,v)$. (Here too, we omit technical details of handling the corner cases.) Let $T'$ be a tree which is passed as an argument to a recursive call in which $u'$ and $v'$ were in $CV_\ell$.
Since $u',v'$ are on $\mathcal{P}_T(u,v)$, tree $T_u \in T' \setminus CV_\ell$ containing $u$ and $T_v \in T' \setminus CV_\ell$ containing $v$ are different. At that point, $u$ (resp. $v$) could have many cut vertices in $\border(u)$ (resp., $\border(v)$). To avoid checking every possible pair of cut vertices in $\border(u)$ and $\border(v)$, we construct another tree, called \emph{contracted tree} which facilitate finding the corresponding pair of cut vertices. 

Fix a vertex $\beta \in \Phi$, corresponding to a recursive call of spanner construction where a tree $(T', rt(T'))$ is passed as an argument, and let $CV_\ell$ denote the set of cut vertices chosen for this level. Furthermore, let $T_1,\ldots, T_p = T' \setminus CV_\ell$ be the subtrees obtained by removing vertices in $CV_\ell$ from $T'$. The set of vertices of the contracted tree $\mathcal{T}_\beta$, corresponding to vertex $\beta$ in $\Phi$, consists of $p$ vertices, $t_1,\ldots,t_p$, corresponding to $T_1,\ldots,T_p$, and $|CV_\ell|$ vertices corresponding to cut vertices in $CV_\ell$. For each vertex $t_i \in \mathcal{T}_\beta$, we add an edge between $t_i$ and all the vertices in $\mathcal{T}_\beta$ corresponding to cut vertices in $\border(T_i)$. Intuitively, the augmented tree $\mathcal{T}_\beta$ identifies every subtree $t_i$ with a single vertex and keeps the tree structure of given tree $T$.

We now explain how $\mathcal{T}_\beta$ facilitates finding cut vertices $u'$ and $v'$ corresponding to vertex $\beta \in \Phi$ which are on $\mathcal{P}_T(u,v)$. First, we find the vertex $t_u$ (resp., $t_v$) in $\mathcal{T}_\beta$ corresponding to subtree $T_u$ (resp., $T_v$) which contains $u$ (resp., $v$). (Here too, we consider then most general case, when neither $u$ nor $v$ are in $CV_\ell$.) Cut vertex $u' \in CV_\ell$ is the first vertex on the path from $t_u$ to $t_v$ in $\mathcal{T}_\beta$. In other words, it can be either parent of $t_u$ or the first child on the path from $t_u$ to $t_v$ in $\mathcal{T}_\beta$. In both cases, $u'$ can be found using level ancestor data structure. 
We can similarly find vertex $v'$. This completes the high-level overview of our navigation algorithm.

\subsubsection{Preprocessing}\label{subsec:preprocess-tree}
\paragraph{Algorithm description.}\label{para:preprocess:alg}
We proceed to give a detailed description of the preprocessing algorithm.
It takes as an input a rooted tree $(T, rt(T))$ which induces a tree metric $M_T$. Notice that $T$ can be transformed in linear time into a pruned tree $(T_\pruned, rt(T_\pruned))$ by invoking the procedure $\Call{Prune}{(T, rt(T)), R(T)}$. Also, any 1-spanner for pruned tree $T_\pruned$ provides a 1-spanner for the original tree $T$ with the same diameter. We may henceforth assume that the original tree $T$ is pruned.

Our preprocessing algorithm construct two types of trees --- augmented recursion trees and contracted trees. We preprocess every such tree in linear time so that subsequent \emph{lowest common ancestor} (henceforth, LCA) and \emph{level ancestor} (henceforth, LA) queries can be answered in constant time. For more details on these algorithms, refer to \cite{BF00, BF04}.

We now give details of the procedure \Call{PreprocessTree}{$(T, rt(T)), R(T), k$}. For pseudocode, see \Cref{alg:preprocess}. This procedure takes as parameters a rooted tree $(T,rt(T))$, the set $R(T)$ of required vertices of $T$, and an integer $k\ge 2$, representing the hop-diameter. It outputs the set of edges of \cite{Sol13} spanner for $T$, together with the augmented recursion tree $(\Phi, rt(\Phi))$. In addition, it creates a data structure $\mathcal{D}_T$ which supports subsequent queries for $k$-hop $1$-spanner paths in $G_T$ between any two vertices $u,v\in R(T)$.

Let $n$ denote the number of required vertices in $T$, that is, $n\coloneqq|R(T)|$. When $n \le k+1$, the algorithm invokes $\Call{HandleBaseCase}{(T, rt(T)), R(T),k}$, which we proceed to describe. If $n = k + 1$ and $rt(T)$ has exactly two children, $u$ and $v$, the edge set of spanner, $E$, consists of the edges of $T$, denoted by $E(T)$, together with edge $(u,v)$; otherwise, it consists of $E(T)$ only. Every vertex $v \in V(T)$ initializes its special adjacency list, $v.adj$ containing only edges in $E$.
The recursion tree returned by this step, $\Phi$, consists of a single vertex $\beta$. For each vertex $v$ in $R(T)$ we create its copy and assign it as an \emph{inner vertex} $v'$ of $\beta$. At this stage, we keep a pointer from $v$ to $v'$ and vice-versa. Keeping inner vertices of each vertex in recursion tree will facilitate answering the queries later on.
The procedure returns $E$ as the edge set of the spanner, together with a (single-vertex) tree $\Phi$ rooted at $\beta$.

In what follows, we consider the case $n > k+1$. The set of cut vertices, $CV_\ell$ is determined by calling the aforementioned procedure $\Call{Decompose}{(T, rt(T)), R(T), \ell}$, with parameter $\ell$ set to $\alpha'_{k-2}(n)$. We create a new vertex $\beta$ and make it a root of the recursion tree $\Phi$. This is done via a call to procedure $\Call{NewVertex}{CV_\ell}$, which assigns to $\beta$ as its inner vertices all the cut vertices in $CV_\ell$. The procedure also keeps track of all the relevant pointers.

The algorithm $\Call{PreprocessTree}{(T, rt(T)), R(T), k}$ returns as its output the set of spanner edges, which consists of edges interconnecting the cut vertices, denoted by $E'$, the edges connecting each cut vertex to required vertices in the tree, denoted by $E''$, and the edges $E_i$, for each of the subtrees $T_i$, $i\in[p]$.
If $k=2$, then there is exactly one vertex $v$ in $CV_\ell$ and we keep $E'$ empty.
For $k=3$ the set $E'$ consists of an edge between every two cut vertices, i.e., $E' = CV_\ell \times CV_\ell$. 
For $k\ge 4$, we first create a tree isomorphic to $T$, which has as its required vertices the inner vertices of $\beta$. The edge set $E'$ is then obtained invoking $\Call{PreprocessTree}{(T', rt(T')), CV_\ell, k-2}$. 

We next compute the trees $T_1,\ldots,T_p$ in $T \setminus CV_\ell$ and the set of edges $E''$ consisting of edges between $c$ and every vertex in $\border(c)$ for all $c \in CV_\ell$.
For each tree $T_i$, we recursively preprocess it by calling $\Call{PreprocessTree}{(T_i, rt(T_i)), R(T)\cap V(T_i), k}$. Let $E_i$ be the edge set returned by this procedure and $(\Phi_i, rt(\Phi_i))$ be the recursion tree for each of the subtrees. We make $rt(\Phi_i)$ a child of $\beta$.

Finally, if $k \ge 3$ we construct the contracted tree $\mathcal{T}_\beta$ which corresponds to $\beta$. This is done via a call to procedure $\Call{CreateContracted}{\beta, \{T_i\}_{i\in[p]}, \{\Phi_i\}_{i \in [p]}}$. This procedure creates representative vertex $t_i$ for each tree $T_i$, $i \in [p]$. In addition, for every inner vertex $c$ of $\beta$ it creates a vertex $c'$ corresponding to it. (Recall that $\beta$ has inner vertices corresponding to vertices $CV_\ell$.) At this stage, it creates a pointer from $c$ to $c'$ and vice versa. At this stage, we have constructed a vertex set for $\mathcal{T}_\beta$; it remains to add the edges to it. For every vertex $c'$ corresponding to inner vertex $c$ of $\beta$, the algorithm connects it to every representative $t_i$ which represents at least one vertex in $\border(c)$. Finally, the root $rt(\mathcal{T}_\beta)$ is the vertex corresponding to the vertex of the lowest level in $T$. The procedure returns $(\mathcal{T}_\beta, rt(\mathcal{T}_\beta))$.

The algorithm returns the set of edges $E' \cup E'' \cup \bigcup_{i\in [p]} E_i$ together with the recursion tree $(\Phi, rt(\Phi))$. The data structure $\mathcal{D}_T$ consists of all the vertices of $T$, all the vertices in every recursion tree, and all the vertices in every contracted tree, together with the data assigned to them. 

\begin{algorithm}
\begin{algorithmic}[1]
\Procedure{PreprocessTree}{$(T, rt(T)), R(T),k$} \Comment{The main algorithm.}
\State\Call{Prune}{$(T,rt(T)), R(T)}$ 
\IfThen {$n \le k+1$}{\Return\Call{HandleBaseCase}{$(T,rt(T)), R(T), k$}} \Comment{$n = |R(T)|$}
\State $CV_\ell \gets \Call{Decompose}{(T, rt(T)), R(T), \ell}$\Comment{$\ell \gets \alpha'_{k-2}(n)$}
\State create $\Phi$ consisting of a single vertex $\beta \gets \Call{NewVertex}{CV_\ell}$
\If {$k=3$} \Comment{When $k=2$, $E'$ is empty.}
\State $E' \gets CV_\ell \times CV_\ell$
\ElsIf {$k\ge 4$}
\State make $T'$, a copy of $T$, with inner vertices of $\beta$ as required vertices\label{line:preprocess:rec-k-start}
\State $(E', (\Phi', rt(\Phi'))) \gets \Call{PreprocessTree}{(T', rt(T')), CV_\ell, k-2}$\label{line:preprocess:rec-k-end}
\EndIf
\State $\{T_1, T_2,\ldots,T_p\}\gets T\setminus CV_\ell$\Comment{root $rt(T_i)$ of $T_i$ is the vertex of the lowest level in $T_i$}
\State $E'' \gets \cup_{u \in CV_\ell} \{u\} \times \border(u)$ 
\For {$i \in [p]$}
\State $(E_i, (\Phi_i, rt(\Phi_i)) \gets \Call{PreprocesTree}{(T_i, rt(T_i)), R(T)\cap V(T_i), k}$\label{line:preprocess:recN}
\State make $rt(\Phi_i)$ child of $\beta$
\EndFor
\IfThen{$k\ge 3$}{$(\mathcal{T}_\beta, rt(\mathcal{T}_\beta)) \gets \Call{CreateContracted}{\beta, \{T_i\}_{i\in[p]}, \{\Phi_i\}_{i\in[p]}}$}
\State \Return $(E' \cup E''\cup \bigcup_{i \in [p]} E_i, (\Phi, \beta))$
\EndProcedure

\Procedure{HandleBaseCase}{$(T, rt(T)), R(T), k$} \Comment{Base case when $n \le k+1$.}
\State $E \gets E(T)$
\IfThen {$n=k+1$ and $rt(T)$ has exactly two children, $u$ and $v$,}{$E \gets E \cup (u,v)$}
\State create $\Phi$ consisting of a single vertex $\beta \gets \Call{NewVertex}{R(T)}$
\State for each $v \in T$, create $v.adj$ based on edges in $E$\label{line:preprocess:adj}
\State\Return $(E, (\Phi,\beta))$
\EndProcedure

\Procedure{NewVertex}{$U$} \Comment{Creates a new vertex for $\Phi$.}
\State create a new vertex $\beta$
\ForAll {$v \in U$}
\State create a copy $v'$ of $v$ and make it inner vertex of $\beta$ 
\State $v.ptr(\Phi) \gets v'$, $v'.ptr(T)\gets v.ptr(T)$, $v'.h \gets \beta$ \label{line:preprocess:pointersNew} \Comment{If $|U|=1$, let $\beta.ptr(T) \gets v.ptr(T)$}
\EndFor
\State \Return $\beta$
\EndProcedure

\Procedure{CreateContracted}{$\beta, \{T_i\}_{i\in[p]}, \{\Phi_i\}_{i\in[p]}$}\Comment{Creates a contracted tree for $\beta \in \Phi$.}
\State $E(\mathcal{T}_\beta) \gets \emptyset$, $V(\mathcal{T}_\beta) \gets \emptyset$
\ForAll {$i \in [p]$} 
\State create a new vertex $t_i$ corresponding to $T_i$, add it to $V(\mathcal{T}_\beta)$
\State $rt(\Phi_i).ptr(\mathcal{T}) \gets t_i$\label{line:preprocess:pointers:calT1}
\EndFor
\ForAll {inner vertices $c$ of $\beta$}
\State create a new vertex $c'$ corresponding to $c$, add it to $V(\mathcal{T}_\beta)$
\State $c.ptr(\mathcal{T})\gets c'$, $c'.ptr(\Phi) \gets c$\label{line:preprocess:pointers:calT2}
\State add to $E(\mathcal{T}_\beta)$ an edge between $c'$ and every $t_i$ such that $\border(c)\cup T_i \neq \emptyset$
\EndFor
\State \Return$(\mathcal{T}_\beta, rt(\mathcal{T}_\beta))$\Comment{$rt(\mathcal{T}_\beta) \in V(\mathcal{T}_\beta)$ corresponds to vertex of the lowest level in $T$.}
\EndProcedure
\end{algorithmic}
\caption{Constructs spanner, together with necessary data required for efficient navigation. Each constructed tree is preprocessed for answering LCA and LA queries.
}\label{alg:preprocess}
\end{algorithm}
An example of the preprocessing algorithm is given in \Cref{fig:recursionTree}. Tree $T$ is depicted on the left. 
Its number of (required) vertices is $n=48$, and it is split into five subtrees $T_1,\dots,T_5$ using the set of four cut vertices, marked green inside of the dotted area. The size of each subtree is at most $\alpha'_{k-2}(n)=\alpha'_2(48)=10$. Tree $T_1$ has size $4 \le k+1$ and it corresponds to base case of the spanner construction. Trees $T_2,\dots,T_5$ have size 10 and are recursively split into subtrees of size at most $\alpha'_2(10)=6$.
Finally, one of the subtrees of $T_3$ (the subtree on the bottom) has size $6$ and gets split using a single cut vertex into two subtrees of size 2 and 3.
The four cut vertices used in the first level of recursion (inside of the dotted region) are interconnected using the construction for $k=2$. Before the spanner construction for $k=2$ is invoked, the algorithm makes the cut vertices required and all the other vertices Steiner vertices and prunes the tree. The pruned tree (as depicted inside of the dotted region) has as its root a Steiner vertex and has four vertices corresponding to the cut vertices; its edges are represented by dashed lines. Since its size is greater than $(k-2)+1=3$, it gets split using the cut vertex (which is the Steiner vertex) into two subtrees, each of size 2.

The recursion tree $\Phi$ corresponding to the spanner construction is depicted on the right towards the bottom. Each non-leaf vertex of $\Phi$ has a contracted tree associated to it. The root of $\Phi_T$, denoted by $\beta$ has a contracted tree $\mathcal{T}_\beta$ associated to it. Its vertices are the cut vertices of $T$ and the vertices $t_1$ to $t_5$ corresponding to subtrees $T_1$ to $T_5$. In the image, the root of $\Phi$ points (via an arrow with a dotted line) to another recursion tree with two vertices, corresponding to the recursive construction for $k=2$.

\begin{figure}
\centering
\input{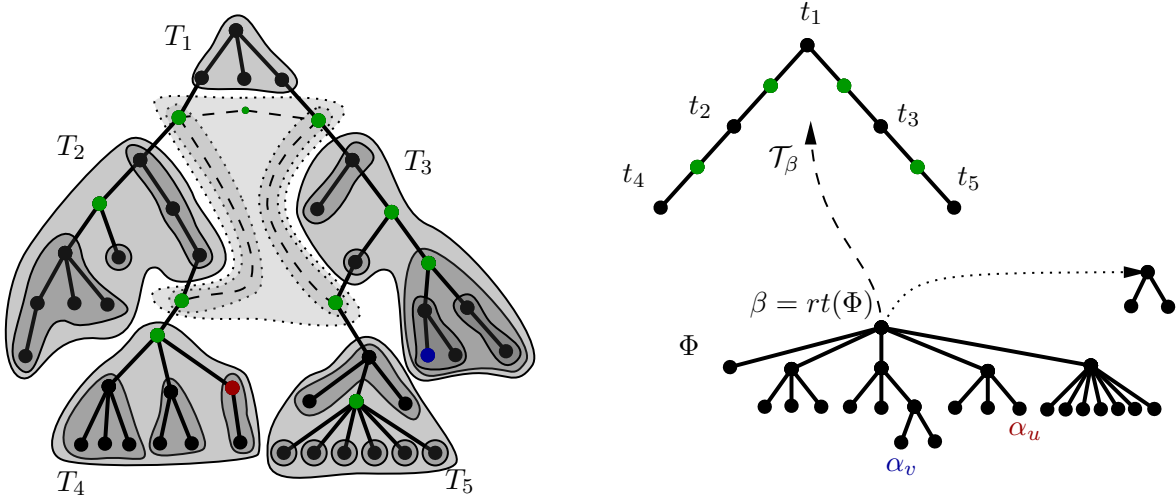}
\caption{Result of the preprocessing algorithm.}
\label{fig:recursionTree}
\end{figure}
\paragraph{Algorithm guarantees.}
We shall use the following lemma from \cite{Sol13} which bounds the number of cut vertices returned by $\Call{Decompose}{(T, rt(T)), R(T), \ell}$.

\begin{lemma}[\cite{Sol13}]\label{lem:CVell-size} Let $n$ be the number of required vertices of the tree, and $\ell = \alpha'_{k-2}(n)$. Then, the size of $CV_\ell$ returned by $\Call{Decompose}{(T, rt(T)), R(T), \ell}$ satisfies: $|CV_{\ell}| = 1$ if  $k=2$, $|CV_{\ell}| \le \sqrt{n}$ if  $k = 3$, and  $|CV_{\ell}| \leq \lfloor \frac{|V(T)|}{\ell+1} \rfloor$ if $k\geq 4$. 
\end{lemma}

In the following lemma, we bound the preprocessing time of the algorithm in \Cref{alg:preprocess}.

\begin{lemma}\label{lm:preporcessingtime} Let $T$ be a tree with  required size $|R(T)| = n$ and $|V(T)| = O(n)$. Algorithm $\Call{PreprocessTree}{(T, rt(T)), R(T),k}$ in time $O(n\alpha_k(n))$ outputs a $k$-hop $1$-spanner $G_T$ for $R(T)$ with $O(n\alpha_k(n))$ edges and the corresponding navigation data structure $\mathcal{D}_T$.
\end{lemma}
\begin{proof}
First, we note that  $G_{T}$ is the same as the $1$-spanner in the construction of Solomon \cite{Sol13}, which was shown to have $O(n\alpha_k(n))$ edges. 	
	
Let $f(n,k)$ be the running time of \Call{PreprocessTree}{$(T, rt(T)), R(T),k$}. Observe that, excluding the recursive calls in lines \ref{line:preprocess:recN} and \ref{line:preprocess:rec-k-end}, the running time is $O(n)$. When $ k =2$, we have that $CV_\ell = 1$ and $|R(T_i)| \leq \alpha'_0(n) = \lceil n/2 \rceil$. Thus, $f(n,2) = \sum_{i\in [p]}f(|R(T_i)|,2) + O(n)$  which resolves to $f(n,2) = O(n\log n) = O(n\alpha_2(n))$. When $k = 3$, we have that $|R(T_i)| \leq \alpha'_1(n)= \lceil \sqrt{n} \rceil$, and hence, $f(n,3) = \sum_{i\in [p]}f(|R(T_i)|,3) + O(n)$  which resolves to $O(n\log\log(n)) = O(n\alpha_3(n))$. When $k\geq 4$, recall that $\ell = \alpha'_{k-2}(n)$, $|R(T_i)| \leq \ell$ and $|CV_{\ell}| \leq \lfloor \frac{n}{\ell+1} \rfloor$ by \Cref{lem:CVell-size}. Thus, we have:
	 	\begin{equation}\label{eq:time-fnk}
		f(n,k) = \sum_{i\in [p]}f(|R(T_i)|,k) + f\left(\frac{n}{\alpha'_{k-2}(n)}, k-2\right) +  O(n)
		\end{equation}
 which resolves to $f(n,k) = O(n\alpha_k(n))$; see Theorem 3.12 in \cite{Sol13} for a detailed inductive proof.
\end{proof}

We next list several properties of the preprocessing algorithm, which follow from the algorithm description. These properties are used in \Cref{subsec:query}, where the query algorithm is presented. Recall that we use $\mathcal{D}_T$ to denote the data structure which consists of all the vertices of $T$, all the vertices in every recursion tree, and all the vertices in every contracted tree, together with the data assigned to them. See \Cref{tab:navData} for the summary of data stored in $\mathcal{D}_T$.

\begin{property}\label{property:lca-la}
Every tree constructed by the algorithm is preprocessed for answering LCA and LA queries in constant time. 
\end{property}

\begin{property}\label{property:T-pointer}
For every vertex $x \in \mathcal{D}_T$, pointer $x.ptr(T)$ is either $\emptyset$ or it points to the vertex corresponding to $x$ in $T$.  In particular, it is defined in the following cases: \emph{(i)} $x \in T$, \emph{(ii)} $x$ is a non-leaf vertex in some augmented recursion tree for construction with $k=2$, \emph{(iii)} $x$ is an inner vertex in some augmented recursion tree.
\end{property}

\begin{property}\label{property:phi-pointer}
For every vertex $x \in \mathcal{D}_T$, pointer $x.ptr(\Phi)$ is either $\emptyset$ or it points to its corresponding inner vertex in some recursion tree. In particular, it is defined in the following cases: \emph{(i)} $x \in T$, \emph{(ii)} $x$ corresponds to a cut vertex in some contracted tree, \emph{(iii)} $x$ is an inner vertex of a non-leaf in some augmented recursion tree for construction with $k\ge 4$.
\end{property}

\begin{property}\label{property:calT-pointer}
For every vertex $x \in \mathcal{D}_T$, pointer $x.ptr(\mathcal{T})$ is either $\emptyset$ or it points to its corresponding vertex in some augmented recursion tree. In particular, it is defined in the following cases: \emph{(i)} $x$ is an inner vertex of a non-leaf in some augmented recursion tree for construction with $k\ge 3$, \emph{(ii)} $x$ is a non-root in some recursion tree for construction with $k \ge 3$.
\end{property}

\begin{property}\label{property:h-pointer}
For every vertex $x \in \mathcal{D}_T$, pointer $x.h$ is either $\emptyset$ or it points to its corresponding vertex in some augmented recursion tree. In particular, it is defined if $x$ is an inner vertex of a vertex in some recursion tree.
\end{property}

\begin{property}\label{property:base-case}
Whenever a vertex $x$ is considered in procedure $\Call{HandleBaseCase}$ with parameter $k$, its entry $x.adj$ contains only adjacency list in the subgraph of $G_T$ induced by vertices considered in the same base case. Moreover, this subgraph contains $O(k)$ vertices.
\end{property}

\begin{property}\label{property:contracted}
Every contracted tree $\mathcal{T}_\beta$, corresponding to a vertex $\beta$ in some recursion tree $\Phi$ satisfies the following: 
\begin{enumerate}[label=(\emph{\roman*})]
\item It only contains representative vertices and cut vertices.
\item There is no edge between two representative vertices.
\item There is an edge between a cut vertex $c$ and a representative vertex $t_i$ iff $c \in \border(T_i)$.
\end{enumerate}
\end{property}

\begin{observation}\label{lem:depth}
Given a tree $T$ of required size $n$ and an integer $k$, the depth of augmented recursion tree corresponding to $T$ is $O(\alpha_k(n))$.
\end{observation}
\begin{proof}
The depth of augmented recursion tree satisfies recurrence $D(n,k) = D(\alpha_{k-2}(n), k) + 1$, with a base case $D(n, k) = 1$, whenever $n \le k+1$. This recurrence has solution $D(n,k) = O(\alpha_{k}(n))$.
\end{proof}

\subsubsection{Query algorithm} \label{subsec:query}

\begin{table}\label{table:data}
\begin{center}
\scalebox{0.9}{
\begin{tabular}{cccccc}
\hline
\multicolumn{1}{|c|}{\multirow{2}{*}{\textbf{field}}} & \multicolumn{1}{c|}{\multirow{2}{*}{\textbf{meaning}}}                           & \multicolumn{4}{c|}{\textbf{defined for vertices in}}                                                                                                                               \\ \cline{3-6} 
\multicolumn{1}{|c|}{}                                & \multicolumn{1}{c|}{}                                                            & \multicolumn{1}{c|}{ $T$} & \multicolumn{1}{c|}{$\Phi$, inner} & \multicolumn{1}{c|}{$\Phi$} & \multicolumn{1}{c|}{$\mathcal{T}$} \\ \hline
\multicolumn{1}{|c|}{$ptr(T)$}                        & \multicolumn{1}{c|}{pointer to  vertex in $T$}                      & \multicolumn{1}{c|}{yes}           & \multicolumn{1}{c|}{yes}        & \multicolumn{1}{c|}{if non-leaf \& $k=2$}                    & \multicolumn{1}{c|}{}                        \\ \hline
\multicolumn{1}{|c|}{$ptr(\Phi)$}                     & \multicolumn{1}{c|}{pointer to  inner vertex in $\Phi$}     & \multicolumn{1}{c|}{yes}           & \multicolumn{1}{c|}{if non-leaf host \& $k\ge 4$}                 & \multicolumn{1}{c|}{}                       & \multicolumn{1}{c|}{if cut vertex}                        \\ \hline
\multicolumn{1}{|c|}{$ptr(\mathcal{T})$}              & \multicolumn{1}{c|}{pointer to  vertex in $\mathcal{T}$}           & \multicolumn{1}{c|}{}              & \multicolumn{1}{c|}{if non-leaf host \& $k\ge 3$}                 & \multicolumn{1}{c|}{if non-root \& $k \ge 3$}                       & \multicolumn{1}{c|}{}                        \\ \hline
\multicolumn{1}{|c|}{$h$}                             & \multicolumn{1}{c|}{pointer to home vertex in $\Phi$}                & \multicolumn{1}{c|}{}              & \multicolumn{1}{c|}{yes}                 & \multicolumn{1}{c|}{}                       & \multicolumn{1}{c|}{}                        \\ \hline
\multicolumn{1}{|c|}{$adj$}                           & \multicolumn{1}{c|}{adjacency table} & \multicolumn{1}{c|}{yes}              & \multicolumn{1}{c|}{yes}                 & \multicolumn{1}{c|}{}                       & \multicolumn{1}{c|}{}                        \\ \hline
\multicolumn{1}{|c|}{$level$}                         & \multicolumn{1}{c|}{level in the tree}                                           & \multicolumn{1}{c|}{yes}           & \multicolumn{1}{c|}{}              & \multicolumn{1}{c|}{yes}                       & \multicolumn{1}{c|}{yes}                     \\ \hline
\multicolumn{1}{l}{}                                  & \multicolumn{1}{l}{}                                                             & \multicolumn{1}{l}{}               & \multicolumn{1}{l}{}                  & \multicolumn{1}{l}{}                        & \multicolumn{1}{l}{}                        
\end{tabular}
}
\end{center}

\caption{Summary of data stored with every vertex in $x \in \mathcal{D}_T$. Undefined entries take value $\emptyset$. Note that the preprocessing algorithm constructs more than one augmented recursion tree when $k\ge 4$ and more than one contracted tree when $k\ge 3$. }\label{tab:navData}
\end{table}

The algorithm which finds a $k$-hop path between $u$ and $v$ in 1-spanner $G_T$ of tree $T$ is presented in \Cref{alg:query}. It takes two vertices $u$ and $v$ and a parameter $k$, representing the hop-diameter of $G_T$.

\paragraph{Algorithm description.}
We proceed to give details of the query algorithm. It takes as an input two vertices $u$ and $v$ and an integer parameter $k \ge 2$, representing the hop-diameter. Let $\Phi$ denote the recursion tree corresponding to spanner construction which considered $u$ and $v$ with parameter $k$.
The algorithm first checks whether $u$ and $v$ were considered in the same base case corresponding to the call to $\Call{HandleBaseCase}$ during the construction of $G_T$. This check is performed in \cref{line:query:check-base}. We check if $u$ and $v$ point to the same leaf in $\Phi$. The inner vertex corresponding to $u$ (resp., $v$) in $\Phi$ obtained via $u.ptr(\Phi)$ (resp., $v.ptr(\Phi)$). We use $u.ptr(\Phi).h$ (resp., $v.ptr(\Phi).h$) to obtain the actual vertex in $\Phi$, which contains $u.ptr(\Phi)$ (resp., $v.ptr(\Phi)$) as its inner vertices. If $u.ptr(\Phi).h$ is equal to $v.ptr(\Phi).h$, the algorithm returns the path found by BFS on the subgraph of spanner $G_T$ induced on all the vertices corresponding to the same base case (\cref{line:query:bfs}). This BFS uses adjacency list $u.adj$ stored with vertex $u$, which contains only the edges of the spanner corresponding to this base case. In other words, the algorithm will only visit the subgraph of $G_T$ induced on the vertices corresponding to the same base case as $u$ and $v$.

When $u$ and $v$ do not correspond to the same vertex in $\Phi$, the algorithm finds LCA in $\Phi$ of $u.ptr(\Phi).h$ and $v.ptr(\Phi).h$, denoted by $\beta$ (\cref{line:query:lca-phi}). 
If $k=2$, then, by \Cref{property:T-pointer}, $\beta$ corresponds to a single vertex in $T$; its corresponding vertex in $T$ is $\beta.ptr(T)$. The algorithm returns path consisting of at most three vertices in $T$, namely $\{u.ptr(T), \beta.ptr(T), v.ptr(T)\}$. We use braces to denote that consecutive duplicates are removed from it. For example, when $u=\beta$, then $u.ptr(T)=\beta.ptr(T)$ and the algorithm returns two vertices: $\{u.ptr(T), v.ptr(T)\}$.

When $k\ge 3$, the algorithm proceeds to find cut vertices corresponding to $u$ and $v$. 
For that purpose, it considers contracted tree $\mathcal{T}_\beta$, corresponding to $\beta$. First of all, it locates vertices corresponding to $u$ (resp., $v$) in $\mathcal{T}_\beta$, via a call to $\Call{LocateContracted}{u, \beta}$.
If $u$ points to $\beta$ in $\Phi$, it means that $u$ is a cut vertex at the required level; all we need to do is to find its corresponding vertex in $\mathcal{T}_\beta$, which is obtained via $u.ptr(\Phi).ptr(\mathcal{T}_\beta)$.  If $u$ is not a cut vertex at the required level, we use level ancestor data structure to find child of $\beta$ on the path to vertex corresponding to $u$. By \Cref{property:calT-pointer}, this child corresponds to a unique vertex $u'$ in $\mathcal{T}_\beta$, which is a representative of connected component containing $u$.  Vertex $v'$ corresponding to $v$ is found analogously.

Next, we would like to find the first cut vertex $x$ on the path from $u'$ to $v'$ (resp., the first cut vertex $y$ on the path from $v'$ to $u'$) in the contracted tree $\mathcal{T}_\beta$. First the algorithm finds lowest common ancestor $c$ of $u'$ and $v'$ (cf.~\cref{line:query:lca-tb}). Then, it invokes $\Call{FindCut}{u,u',v',\beta,c}$, which we explain next. If $u'$ already corresponds to a cut vertex, we assign $u'$ to $x$. If that is not the case, when $v'$ is a descendant of $u'$, we let $x$ be the child of $u'$ on the path to $v'$ and otherwise we let it be the parent of $u'$; in both cases, we can find $x$ using level ancestor data structure on $\mathcal{T}_\beta$. Vertex $y$ is found similarly, using a call to $\Call{FindCut}{v,v',u',\beta,c}$.
When $k=3$ the algorithm reports vertices corresponding to $u$, $x$, $y$, $v$ in $T$ (\cref{line:query:k3}). Otherwise, it proceeds recursively to find a $(k-2)$-hop path between inner vertices of $\beta$ in $\Phi$ corresponding to $x$ and $y$ (\cref{line:query:recurse}). 

We refer reader to \Cref{fig:recursionTree} for an illustration of a query algorithm.
Upon a query to navigate between the red vertex $u \in T$ and the blue vertex $v \in T$, the algorithm uses $u.ptr(\Phi)$ and $v.ptr(\Phi)$ to find the corresponding vertices in $\Phi_T$, denoted by $\alpha_u$ and $\alpha_v$. Since $\alpha_u \neq \alpha_v$, this means that $u$ and $v$ were not considered together in a base case (corresponding to invocation of $\Call{HandleBaseCase}$). Next, the algorithm finds $LCA(\alpha_u, \alpha_v)$, which is the root $rt(\Phi)$, denoted by $\beta$ in the picture. Using the level ancestor data structure, we find children $\alpha_1,\alpha_2$ of $\beta$, on the path to $\alpha_u$ and $\alpha_v$, respectively. Vertex $\alpha_1$ points to $t_4$ in the contracted tree $\mathcal{T}_\beta$; similarly, vertex $\alpha_2$ points to $t_3$. The LCA of $t_3$ and $t_4$ is the root of $\mathcal{T}_\beta$ and the cut vertices corresponding to $t_3$ and $t_4$ are their parents. Finally, we recursively find the path (of at most 2 hops) between the chosen cut vertices using the augmented recursion tree for construction with $k-2$ (pointed to by a dotted arrow).

\begin{algorithm}
\begin{algorithmic}[1]
\Procedure{FindPath}{$u, v, k$} \Comment{The query algorithm.}
\If{$u.ptr(\Phi).h = v.ptr(\Phi).h$ and $u.ptr(\Phi).h$ is a leaf of $\Phi$}\label{line:query:check-base} 
\State\Return\Call{BFS}{$u,v$}\label{line:query:bfs} \Comment{BFS on $G_T$ induced on $\{w \in V(T) \mid w.ptr(\Phi).h = u.ptr(\Phi).h\}$.} 
\EndIf
\State $\beta \gets \textsc{LCA}(u.ptr(\Phi).h, v.ptr(\Phi).h)$\label{line:query:lca-phi}
\If{$k=2$}
\State\Return $\{u.ptr(T), \beta.ptr(T), v.ptr(T)\}$\label{line:query:k2}
\EndIf
\State$u' \gets \Call{LocateContracted}{u, \beta}$
\State$v' \gets \Call{LocateContracted}{v, \beta}$
\State $c \gets \textsc{LCA}(u', v')$\label{line:query:lca-tb}
\State $x \gets \Call{FindCut}{u, u', v', \beta, c}$
\State $y \gets \Call{FindCut}{v, v', u', \beta, c}$
\If {$k=3$}
\State\Return $\{u.ptr(T), x.ptr(\Phi).ptr(T), y.ptr(\Phi).ptr(T), v.ptr(T)\}$\label{line:query:k3}
\Else
\State\Return $\{u.ptr(T), \Call{FindPath}{x.ptr(\Phi), y.ptr(\Phi), k-2}, v.ptr(T)\}$\label{line:query:recurse}
\EndIf
\EndProcedure
\State
\Procedure{LocateContracted}{$u, \beta$} \Comment{Locates $u'$ corresponding to $u$ in $\mathcal{T}_\beta$.}
\If {$u.ptr(\Phi).h = \beta$}
\State\Return $u.ptr(\Phi).ptr(\mathcal{T}_\beta)$
\Else
\State \Return$\textsc{LA}(u.ptr(\Phi).h, \beta.level+1).ptr(\mathcal{T}_\beta)$ 
\EndIf
\EndProcedure
\State
\Procedure{FindCut}{$u, u', v', \beta, c$} \Comment{Finds the first cut vertex $x$ on the path from $u'$ to $v'$.}
\If {$u.ptr(\Phi).h = \beta$}
\State \Return$u'$
\ElsIf {$u' = c$}
\State \Return$\Call{LA}{v', u'.level+1}$
\Else 
\State \Return$\Call{LA}{u', u'.level-1}$
\EndIf
\EndProcedure
\end{algorithmic}
\caption{Query for a $k$-hop path in tree 1-spanner $G_T$ between two vertices $u$ and $v$.}
\label{alg:query}
\end{algorithm}

\paragraph{Algorithm guarantees.} We next argue the correctness of the query algorithm.
\begin{lemma}\label{lem:query:correct}

Given a tree $T$ preprocessed by $\Call{PreprocessTree}{(T, rt(T)), R(T), k}$, and two vertices $u,v \in R(T)$, the algorithm $\Call{FindPath}{u,v,k}$ outputs a 1-spanner path between $u$ and $v$ in $G_T$ consisting of at most $k$ edges.
\end{lemma}
\begin{proof}
We will prove the lemma by structural induction.

The first base case is when the condition in \cref{line:query:check-base} is true and the algorithm uses BFS to report the path (\cref{line:query:bfs}).
By \Cref{property:phi-pointer,property:h-pointer}, we know that $u.ptr(\Phi).h$ and $v.ptr(\Phi).h$ point to vertices corresponding to $u$ and $v$ in $\Phi$. Moreover, since they correspond to a leaf in $\Phi$, this means that $u$ and $v$ were in the same tree $T_{base}$ processed by $\Call{HandleBaseCase}{(T_{base}, rt(T_{base})), R(T_{base}), k}$. By \Cref{property:base-case}, we know that vertex $u$ contains adjacency list $v.adj$ restricted to $G_T$ induced on $V(T_{base})$; same holds for $v$ and every other vertex in $V(T_{base})$.
This ensures that BFS will find the shortest path from $u$ to $v$, which is guaranteed in \cite{Sol13} to be the 1-spanner path of at most $k$ edges. 

The second base case is when $k=2$. The algorithm returns path $\{u.ptr(T),\beta.ptr(T),v.ptr(T)\}$, which, by \Cref{property:T-pointer} correctly map to corresponding vertices in a given tree $T$. We next argue that this path is a valid 1-spanner path in $G_T$. Among the common vertices on the paths from $rt(\Phi)$ to $u.ptr(\Phi).h$ and from $rt(\Phi)$ to $u.ptr(\Phi).h$, vertex $\beta$ has the lowest level. 
After removing $\beta.ptr(T)$ from $T$, vertices $u$ and $v$ are not in the same subtree. In other words, $\beta.ptr(T)$ is on the path $\mathcal{P}_T(u,v)$.
Suppose that $u \neq \beta.ptr(T)$ and $v \neq \beta.ptr(T)$ and let $T_u$ (resp. $T_v$) be the tree containing $u$ (resp., $v$) after $\beta.ptr(T)$ has been removed. By the previous argument, it must be that $T_u \neq T_v$. From the inductive statement, we know that $u.ptr(T)$ (resp., $v.ptr(T)$) is a required vertex in $T_u$ (resp., $T_v$). This means that $G_T$ contains an edge between $(u.ptr(T), \beta.ptr(T))$ and $(v.ptr(T), \beta.ptr(T))$. Hence, the path $(u, \beta.ptr(T), v.ptr(T))$ exists in $G_T$ and is a valid 1-spanner path in $T$. Cases when $u=v$, $u = \beta.ptr(T)$, and $v = \beta.ptr(T)$ are handled similarly.

The third base case is when $k=3$. The algorithm first finds vertices $u'$ and $v'$ corresponding to $u$ and $v$ in the contracted tree $\mathcal{T}_\beta$. The correctness follows by \Cref{property:calT-pointer}. Finding relevant cut vertex for $u'$ is done in procedure $\Call{FindCut}{u, u', v', b, c}$. Correctness follows by \Cref{property:contracted}. Hence, $x$ (resp., $y$) is the closest cut vertex to $u$ (resp. $v$) on $\mathcal{P}_T(u,v)$. Recall that by inductive statement, $u$ and $v$ are required vertices. 
Since $x$ is in $\border(u)$ and similarly $y$ in $\border(v)$, there are edges $(x,u)$ and $(y,v)$ in $G_T$. In addition, since $x$ and $y$ are cut vertices, they are connected via an edge. This concludes the analysis of base cases.

We take $k\ge 4$ for an inductive step. From the analysis of the base case $k=3$, we know that cut vertices $x$ and $y$ in $\mathcal{T}_\beta$ are correctly computed, and $G_T$ contains edges $(u, x)$ and $(y,v)$. By \Cref{property:phi-pointer} the pointers $x.ptr(\Phi), y.ptr(\Phi)$ point to inner vertices of $\beta$ in $\Phi$. Since $x.ptr(\Phi)$ and $y.ptr(\Phi)$ are required vertices for $\Phi'$ corresponding to construction with hop-diameter $k-2$ (cf. \crefrange{line:preprocess:rec-k-start}{line:preprocess:rec-k-end} in $\Call{PreprocessTree}$), by the inductive hypothesis the recursive call in \cref{line:query:recurse} returns a 1-spanner path between $x.ptr(\Phi)$ and $y.ptr(\Phi)$ of at most $k-2$ hops. Recalling that we have at most two more hops, i.e., if $u \notin CV_\ell$, $(u.ptr(T),x.ptr(\Phi).ptr(T))$ and if $v \notin CV_\ell$, $(y.ptr(\Phi).ptr(T), v.ptr(T))$, the statement follows.
\end{proof}

The following lemma states the running time of our query algorithm.
\begin{lemma}\label{lem:query:time}
Algorithm $\Call{FindPath}{u,v,k}$ runs in time $O(k)$.
\end{lemma}
\begin{proof}
Procedures $\Call{LocateContracted}{}$ and $\Call{FindCut}{}$ access pointer values and perform LCA and LA queries; the number of such operations is constant, allowing us to conclude that both procedures run in $O(1)$ time. The only nontrivial operation in $\Call{FindPath}{}$ is performing BFS in \cref{line:query:bfs}. By \Cref{property:base-case}, this BFS visits only the subgraph of $G_T$ induced on the vertices which correspond to the same base case as $u$ and $v$. The number of vertices in this subgraph is $O(k)$, hence the running time of BFS is also $O(k)$.

In conclusion, the algorithm either performs $O(k)$ operations and does not continue recursively, or it performs constant number of operations and proceeds recursively with parameter $k-2$. This allows us to conclude that the running time of $\Call{FindPath}{}$ is $O(k)$.
\end{proof}

\subsection{Navigating tree covers}\label{subsec:treeCoverNav}
To prove \Cref{thm:oracle}, we rely on tree cover theorems summarized in~\Cref{tab:treeCover}. Let $\zeta$ denote the number of trees in the cover and $\gamma$ the stretch of the cover; let $M_X=(X, \delta_X)$ be the metric space we are working on.
For each of the $\zeta$ trees in the cover, we employ \Cref{thm:treeNavigate} and construct a spanner $G_{T_i}$ and a data structure $\mathcal{D}_{T_i}$. 
For Ramsey tree covers, in the preprocessing step we store a mapping from every point $x$ in the metric space to its ``home'' tree.
Upon a query for a path between $u$ and $v$, for the Ramsey tree covers it is sufficient to use this information to find the corresponding tree in constant time. Otherwise, for each of the $\zeta$ trees, we query $\mathcal{D}_{T_i}$ for the distance between $u$ and $v$ in $T_i$ in $O(1)$ time. (This step takes $O(\zeta)$ time.)
Once the tree with the smallest distance between $u$ and $v$, $T^*$, has been found, we query for the $k$-hop shortest path in $T^*$ between $u$ and $v$ in $O(k)$ time using the result from \Cref{thm:treeNavigate}.

\section{Fault tolerance in doubling metrics} \label{ft}

In this section we strengthen the navigation scheme of \Cref{SectionPathOracle} to achieve fault-tolerance in doubling metrics.
We start with required definitions, move on to presenting a new construction of tree covers in doubling metrics, 
and then build on this tree cover to get a fault-tolerant (FT) spanner of bounded hop-diameter. Equipped with such a spanner, obtaining an FT navigation scheme follows along similar lines to the one presented \Cref{SectionPathOracle} for a non-FT spanner.
 
Let $X= (X,\delta_X)$ be an $n$-point metric of doubling dimension $d$. An \emph{$f$-fault-tolerant (FT)} $t$-spanner of $X$ is a $t$-spanner for $X$ such that, for every set $F\subseteq X$ of size at most $f, f \le n-2$, called a \emph{faulty set}, it holds that
 $\delta_{H\setminus F}(x,y)\leq  t \cdot \delta_{X}(x,y)$, for any pair of $x, y \in X \setminus F$.
An $f$-FT spanner $H$ is said to have {\em hop-diameter} $k$ if the hop-diameter of $H\setminus F$ is at most $k$ for any faulty set $F$ of size at most $f$, and thus there is a non-faulty $t$-spanner path in $H$ of at most $k$ hops for any pair of non-faulty points. The main result of this section is a construction of an $f$-FT spanner with bounded hop-diameter for doubling metrics whose size matches the size bound for non-FT spanners (up to the dependency on $f$). Our construction relies on the notion of a \emph{robust tree cover} that we introduce below. This new tree cover notion generalizes the Euclidean ``Dumbbell Tree'' theorem (Theorem 2 in \cite{ADMSS95}). 
In what follows, we shall use $\mathcal{P}_T(u,v)$ to denote the path between leaves $x$ and $y$ in a rooted tree $T$. We denote by $T_v$ a subtree of $T$ rooted at a vertex $v\in T$.

\begin{definition}[Robust Tree Cover]\label{def:robust-TreeCover} A robust $(\gamma,\zeta)$-tree cover $\mathcal{T}$ for a metric $(X,\delta_X)$ is a collection of $\zeta$ trees satisfying:
	\begin{enumerate}
		\item[(1)] For every tree $T\in \mathcal{T}$, there is a 1-to-1 correspondence between points in $X$ and leaves of $T$.
		\item[(2)] For every $x \not= y \in X$, there exists a tree $T\in \mathcal{T}$ such that the path from $x$ to $y$ obtained by replacing each vertex $v$ in $\mathcal{P}_T(x,y)$ with an \emph{arbitrary} leaf point of $T_v$ has weight at most $\gamma\cdot \delta_X(x,y)$. We say that $T$ \emph{covers} $x$ and $y$. 
	\end{enumerate}
\end{definition}

Property (2) in \Cref{def:robust-TreeCover}, which we call \emph{robustness}, implies  that we can obtain an (ordinary) tree cover by replacing any internal vertex of a tree $T\in \mathcal{T}$ with a point associated with an arbitrary leaf in the subtree rooted at that vertex. The robustness is the key in our construction of an $f$-FT spanner with a bounded hop-diameter. In the following theorem, we show that doubling metrics have robust tree covers with a few trees; the proof is deferred to \Cref{subsec:const-TreeCover}.

\begin{theorem}\label{lm:robust-treecover} For any metric $(X,\delta_X)$ of doubling dimension $d$ and any parameter $\eps > 0$, we can construct a robust $(1+\eps,\eps^{-O(d)})$-tree cover $\mathcal{T}$ for $(X,\delta_X)$ in $O_{d,\eps}(n\log{n})$ time. 
\end{theorem}

The $O_{d,\eps}$ notation hides the dependency on $d$ and $\eps$.
The  tree cover theorem by~\cite{BFN19B} for doubling metrics generalizes all but the robustness of the dumbbell tree theorem. By examining the proof closely, we observe that, in the tree cover of~\cite{BFN19B}, each internal vertex of the tree is replaced by a \emph{specific point} chosen from the leaves in the subtree rooted at that vertex; in particular, Claim 8 in \cite{BFN19B} fails if the point is chosen arbitrarily from the leaves.  

In the following, we show how to construct a FT spanner with  a bounded hop-diameter from a robust tree cover.

\subsection{Construction of fault-tolerant spanners with bounded hop-diameter} \label{sec:ft-spanner}

\begin{theorem}
Given an $n$-point metric $(X,\delta_X)$ of doubling dimension $d$, a parameter $\eps > 0$, and integers $1\le f\le n-2$,  and $k \geq 2$, we can construct an $f$-FT spanner with hop-diameter $k$ and $\eps^{-O(d)}nf^2\alpha_k(n)$ edges in  $O_{d,\eps}(n(\log(n) + f^2 \alpha_{k}(n)))$ time.
\end{theorem}
\begin{proof}
Let $\mathcal{T}$ be a robust $(1+\eps,\eps^{-O(d)})$-tree cover constructed as in \Cref{lm:robust-treecover}. For each tree $T\in \mathcal{T}$, we construct a graph $H_T$ and then form an $f$-FT spanner $H$  as
		$H = \cup_{T\in \mathcal{T}} H_T$.

Initially, the vertex set of $H_T$ contains points in $X$, and the edge set of $H_T$ is empty.
 We then construct a $1$-spanner for $T$ with $k$ hops and $O(n\alpha_k(n))$ edges in $O(n\alpha_{k}(n))$ time, denoted by $K_T$, using the algorithm of Solomon~\cite{Sol13}. Note that edges in $T$ are unweighted. For every vertex $v\in T$, we choose a set $R(v)$ of (arbitrary) $f+1$ points associated with leaves of $T_v$; if $T_v$ has strictly less than $f+1$ leaves, $R(v)$ includes all the leaves. For every edge $(u,v) \in K_T$, we add to $H_T$ edges between points in $R(u)$ and $R(v)$ to make a biclique. The weight of each edge is the distance between its endpoints in $X$. This completes the construction of $H_T$ and hence of $H$.
 
 Observe by the construction that $|E(H_T)| = O(f^2 |E(K_T)|) = O(f^2 n \alpha_k(n))$. It then follows that $|E(H)| = |\mathcal{T}| O(f^2 n \alpha_k(n)) = \eps^{-O(d)}nf^2\alpha_k(n)$. Observe also by the construction that the running time to construct $H_T$ is $O(f^2 n\alpha_k(n))$. Thus, the running time to construct $H$ is $O_{d,\eps}(n\log(n))+ O(f^2 n \alpha_{k}(n))$, as claimed.
 
 Finally, we bound the stretch and the hop-diameter of $H$. Let $x\not= y$ be any two non-faulty points in $X$, and $T$ be a tree in $\mathcal{T}$ that covers $x$ and $y$.  Let $Q$ be any $k$-hop $1$-spanner path between $x$ and $y$ in $K_T$. Let $x = v_0, \ldots, v_{k} = y$ be vertices of $Q$.  We claim that for every $i \in [k]$, there exists a non-faulty point in $R(v_i)$.  If $|R(v_i)| = f+1$, then clearly it contains a non-faulty point. Otherwise, $R(v_i)\cap \{x,y\}\not=\emptyset$. This is because  $Q$ is a $1$-spanner path and hence, any vertex in $Q$ is either an ancestor of $x$ or an ancestor of $y$ or both.  
 
We now construct a $k$-hop path $P$ for $Q$ as follows. For every $i \in [k]$, we replace $v_i$ by a non-faulty point $p_i\in R(v_i)$. 
Thus, $P$  is a path in $H_T$ (and hence in $H$) of hop-diameter $k$. Furthermore, by property (2) in \Cref{def:robust-TreeCover} and the fact that $Q$ is a $1$-spanner path, $P$ has stretch $(1+\eps)$, as desired. 
\end{proof}

\subsection{Construction of Robust Tree Covers}\label{subsec:const-TreeCover}

In this section, we prove \Cref{lm:robust-treecover}. Our construction follows the construction of tree covers of Bartal et al.~\cite{BFN19B}. An \emph{$r$-net} of a metric space $(X,\delta_{X})$ is a subset of points $N\subseteq X$ such that (a) for every two different points $x\not= y\in N$, $\delta_{X}(x,y)> r$ and (b) for every point $x\in X$, there exists a point $y\in N$ such that  $\delta_{X}(x,y)\leq r$. We introduce the notion of \emph{pairing} cover for nets (formally defined in \Cref{def:pairing}), which is the key to achieving the robustness of our tree cover. We first review the construction of  Bartal et al.~\cite{BFN19B}, and then describe how the pairing cover can be used to construct a robust tree cover. 

The construction of  Bartal et al.~\cite{BFN19B} can be  divided into two steps.

(Step 1) They consider a hierarchy of nets $N_0\supseteq N_1 \supseteq N_2 \supseteq \ldots$, where $N_i$ is a $2^i$-net of $(X,\delta_X)$.\footnote{We chose indices to start from 0 for the ease of presentation; cf.~\cite{BFN19B}.}
Each net $N_i$ is then \emph{partitioned} into $\sigma = \eps^{-O(d)}$  \emph{well-separated} sets $N_{i1},\ldots, N_{i\sigma}$ in the sense that for every $x\not=y \in N_{it}$, $\delta_X(x,y) = \Omega(2^i/\eps)$ for any $t\in \{1,\ldots,\sigma\}$. 

(Step 2) They construct a collection of $O(\sigma \log(1/\eps))$ trees $\{T_{j,p}\}$ where $j\in \{1,2,\ldots, \sigma\}$ and $p \in \{0,1,\ldots, \log(1/\epsilon)-1\}$. Each tree $T_{j,p}$ is constructed by considering levels $i$ of the net hierarchy such that $i \equiv p \mod \log(1/\eps)$ and marking points as \emph{clustered} along the way. Specifically, for every point $x\in N_{ij}$ that is unclustered, add all unclustered points at distance $O(2^i/\eps)$ from $x$ to the tree rooted at $x$ as the children of $x$; these points are then marked as \emph{clustered}. 

It follows from the construction that every internal node of each tree is associated with a unique point $x \in X$. To achieve the robustness, we modify the construction of  Bartal et al.~\cite{BFN19B} in two ways. In Step 1, we construct a \emph{cover} (instead of a partition) of size $\eps^{-O(d)}$ for $N_i$ that has a \emph{pairing property}: each point $x$ in a set  $C_{ij}$ in the cover has at most one point $y \in C_{ij}$ such that $\delta_X(x,y)\leq 2^i/\eps$; $y$ is said to be \emph{paired} with $x$. (See \Cref{def:pairing} for a formal definition.)  In Step 2, for each point $x \in C_{ij}$, we connect the subtree containing $x$ and all other subtrees containing vertices within distances $O(2^i)$ from $x$, and the subtree containing $y$ where $y$ is paired with $x$ in $C_{ij}$.

We now give the details of the construction of a robust tree cover. In this section, our focus is primarily on describing the algorithms and proving various properties of the cover. The implementation is discussed in \Cref{subsec:implement-RTC}.  We say that a collection of subsets $\mathcal{C}$ of a set $S$ is a \emph{cover} for $S$ if $\cup_{C\in \mathcal{C}} C = S$. 

\begin{definition}[Pairing Cover]\label{def:pairing} A cover $\mathcal{C}_i$ of a $2^i$-net $N_i$ is a pairing cover if:
	\begin{itemize}[noitemsep]
		\item[(1)] For every set $C\in \mathcal{C}_i$ and every $x\in C$, there exists at most one point $y \neq x$ in $C$ such that $\delta_X(x,y)\leq 2^i/\eps$.
		\item[(2)] For every $x\not= y \in N_i$ such that $\delta_X(x,y)\leq 2^i/\eps$, there exists a set $C\in \mathcal{C}_i$ such that both $x,y$ are in $C$.  We say that $x$ and $y$ are paired by $C$.
	\end{itemize}
\end{definition}

Next, we construct a pairing cover for  $N_i$ with a small number of sets. We use the following well-known packing lemma.

\begin{lemma}[Packing Lemma] \label{lm:packing-dd} Let $P$ be a point set in a metric $(X,\delta_X)$ of doubling dimension $d$ such that for every $x\not=  y \in P$, $r < \delta_X(x,y) \leq R$. Then $|P|\leq \left( \frac{4R}{r}\right)^d$. 
\end{lemma}

\paragraph{Step 1: Constructing a pairing cover of $N_i$.~} The construction has two smaller steps. First, we construct a well-separated partition $\mathcal{P}_i$ of $N_i$ following Bartal et al.~\cite{BFN19B}. Then in the second step, we construct a pairing cover $\mathcal{C}_i$ from $\mathcal{P}_i$.

\begin{itemize}
	\item \textbf{Step 1a.~} Initially $\mathcal{P}_i = \emptyset$. We consider each point $x\in N_i$ in turn, and if there exists a set $P \in \mathcal{P}_i$ such that $\delta_X(x,y) >  (3/\eps)2^i$ for every $y\in P$, then we add $x$ to $P$. Otherwise, we add a new set $\{x\}$ to $\mathcal{P}_i$. Let $\sigma_1 = |\mathcal{P}_i|$. 
	\item \textbf{Step 1b.~}  Let $\sigma_2 = \max_{x\in N_i}|\{y \in N_i: \delta_X(x,y)\leq 2^i/\eps\}|$.  For each set $P \in \mathcal{P}_i$, we construct a collection $\mathcal{C}(P) = \{P_1,\ldots,P_{\sigma_2}\}$ of $\sigma_2$ sets as follows. For each $x\in P$, let $\langle y_1,y_2,\ldots, y_{\sigma_2}\rangle$ be a sequence of all points (in arbitrary order)  in $N_i$ that have distances at most $2^i/\eps$ from $x$. (Possibly, there could be strictly less than $\sigma_2$ such points, and in this case, we duplicate some points to get exactly $\sigma_2$ points in the sequence.) We then construct the set $P_j = \cup_{x\in P}\{x,y_j\}$ for each $j \in \{1,2,\ldots,\sigma_2\}$.  That is, $P_j$ contains every point $x$ in $P$ and the $j$-th point in its sequence. Finally, we set:
	$\mathcal{C}_i = \cup_{P\in \mathcal{P}_i}\mathcal{C}(P)$.
\end{itemize}

\begin{figure}[H]
\centering
\input{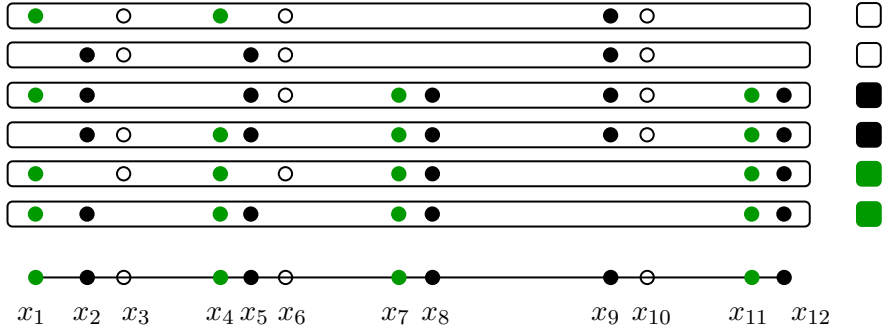}
\caption{Pairing cover of a net.}
\label{fig:pairingCover}
\end{figure}
An example of a pairing cover is presented in \Cref{fig:pairingCover}. Fix level $i=0$ for simplicity. We have a net $N$ of 12 points on a line, $x_1, \ldots, x_{12}$. The partition $\mathcal{P}$ consists of $\sigma_1 = 3$ sets: green, black and white points. For the set of green points, we construct two sets. In the first set, we add $x_1$ and its closest point within distance $1/\eps$, $x_2$, and similarly $\{x_4, x_5\}$, $\{x_7, x_8\}$, and $\{x_{11}, x_{12}\}$. In the second set, for every green point, we add it together with its second-closest point within distance $1/\eps$.

In the following lemma, we show that $\mathcal{C}_i$ is a paring cover of $N_i$.

\begin{lemma}\label{lm:pairing-cover} $\mathcal{C}_i$ is a pairing cover of $N_i$ of size $\eps^{-O(d)}$.
\end{lemma}
\begin{proof}[Proof of \Cref{lm:pairing-cover}]
	First, we bound the size of $\mathcal{C}_i$. Observe by the construction that $|\mathcal{C}_i| = \sigma_1\cdot \sigma_2$. In what follows, we show that $\sigma_1 = \eps^{-O(d)}$ and $\sigma_2 = \eps^{-O(d)}$, implying the bound on $|\mathcal{C}_i|$ as claimed in the lemma. 
	
	Observe that we add a new set $\{x\}$ to $\mathcal{P}_i$ in Step 1a  if every point  in $N_i$ considered before $x$ and  within a distance of  $(3/\eps)2^{i}$ from $x$  belongs to an existing partition of $\mathcal{P}_i$. By the packing lemma (\Cref{lm:packing-dd}), there are  at most $\tau = \eps^{-O(d)}$ such points. Once we make $\tau + 1$ sets, the algorithm will not add any new set. It follows that  $\sigma_1 = \tau + 1 =  \eps^{-O(d)}$. Also by the packing lemma, since $N_i$ is a $2^i$-net,  $\sigma_2 = \eps^{-O(d)}$, as claimed.
	
	Next, we show the pairing property of $\mathcal{C}_i$. Observe by construction that $\mathcal{C}_i$ is a cover of $N_i$. Furthermore, for every $C\in \mathcal{C}_i$, $C = P_j$ for some $P_j \in \mathcal{C}(P)$ constructed from a set $P \in \mathcal{P}_i$, for some $j \in [\sigma_2]$. Recall by the construction in Step 1a that the distance between every two points in $P$ is at least $(3/\eps)2^i$.  It follows that for every $x\in C$, if there exists $y \in C$ such that $y\not= x$ and $\delta_X(x,y)\leq 2^i/\eps$, then either (a) $y = x_j$ where $x_j$ is the $j$-th point in the sequence of $x$ and $x\in P$ or  (b) $x = y_j$ where $y_j$ is the $j$-th point in the sequence of $y$ and $y\in P$.  In either case, there is only one such $y$.  
	
	Finally, consider any pair of points $x\not=y \in N_i$ such that $\delta_X(x,y)\leq 2^i/\eps$. Since $\mathcal{P}_i$ is a partition of $N_i$, there exists a set $P\in \mathcal{P}_i$ such that $x\in P$. Since $\delta_X(x,y)\leq 2^i/\eps$, $y$ must be some point $x_j$ in the sequence of $x$ for some $j\in \{1,2,\ldots,\sigma\}$. Thus, $\{x,y\}\in P_j$. That is, there is a set in $\mathcal{C}_i$ that contains both $x$ and $y$.
\end{proof}

\paragraph{Step 2: Constructing a robust tree cover $\mathcal{T}$.~} Let $N_0\supseteq N_1\supseteq \ldots$  be the hierarchy of nets of $(X,\delta_X)$ where $N_i$ is a $2^i$-net of $X$  and the last net in the sequence contains a single point. By scaling, we assume that the minimum distance in $X$ is larger than $1/(4\eps)$. For each net $N_i$, we construct a pairing cover $\mathcal{C}_i$, and form a sequence $\langle C_1,C_2,\ldots, C_{\sigma_3} \rangle$ of sets in $\mathcal{C}_i$; here $\sigma_3$ is the size of $\mathcal{C}_i$, which is $\eps^{-O(d)}$ by \Cref{lm:pairing-cover}. For each $j \in \{1,2,\ldots, \sigma_3\}$  and each $p \in \{0,1,\ldots, \lceil \log(1/\eps)\rceil - 1\}$, we construct a tree $T_{j,p}$ and form the cover:
$\mathcal{T} = \{T_{j,p}:  j \in \{1,2,\ldots, \sigma_3\} \wedge  p \in \{0,1,\ldots, \lceil \log(1/\eps)\rceil - 1\}\}$.
 Clearly the size of the cover is $O(\sigma_3 \log(1/\eps)) = \eps^{-O(d)}$. 

We now focus on constructing $T_{j,p}$; the construction is  in a bottom-up manner as follows. $T_{j,p}$ has $n$ leaves which are in 1-to-1 correspondence with points in $X$. Let $I = \{i: i\equiv p \mod \lceil \log(1/\eps)\rceil\}$ be the set of levels congruent to $p$ modulo $\lceil \log(1/\eps)\rceil$. For each level $i \in I$ from lower levels to higher levels, let $C_j$ be the $j$-th set in the sequence of the pairing cover $\mathcal{C}_i$. Let $i' = i - \lceil \log(1/\eps)\rceil$, and $F_{i'}$ be the collection of trees constructed at level $i'$. ($F_0$ contains leaves of $T$.)  For each point $x\in C_j$, let $T_x \in F_{i'}$ be the  tree containing $x$, and $F_x\subseteq F_{i'}$ be a collection of subtrees such that each tree $T\in F_x$ contains a point $z$ within distance $2^i$ from $x$.  

For every two points $x,y$ that are paired by $C_j$, we add a new node $v$ and make the roots of trees in $\{T_x,T_y\}\cup F_x \cup F_y$ children of $v$.  The resulting forest after this process is denoted by $F_i$.  
At the top level $i_{\max}$, if $F_{i_{\max}}$ contains more than one tree, we merge them into a single tree by creating a new node, and making the roots of the trees in $F_{i_{\max}}$ children of the new node. The resulting tree is $T_{j,p}$, and this completes the construction of Step 2.

In the following lemma, we argue that $F_i$ is a forest and bound the diameter of trees in $F_i$.

\begin{lemma}\label{lm:diam-Fi} For every level $i$ the following statements are true:
\begin{enumerate}[label=(\roman*)]
\item $F_i$ is a forest.
\item Let $T$ be a tree in $F_i$, and $\diam(T)$ be the diameter of the set of points associated with leaves of $T$. Then $\diam(T) \leq (1/\eps + 20)2^i$ when $\eps \leq 1/12$.
\end{enumerate}
\end{lemma} 
\begin{proof}[Proof of \Cref{lm:diam-Fi}] We prove the lemma by induction. The statement is vacuously true for $F_0$. We assume that the statement is true for $F_{i'}$, where $i' = i-\lceil \log(1/\eps) \rceil$, and prove it for $F_i$.

\begin{enumerate}[label=(\roman*)]
\item Let $x'$ be a point in $C_j \setminus \{x,y\}$. Recall that by construction $\delta_X(x,x') > 2^i/\eps$. It suffices to argue that: (a) $T_x$ cannot contain $x'$ and (b) no tree in $F_x$ can be in $F_{x'}$. For (a) we use induction hypothesis and observe that for any other $z \in T_x$, we have $\delta_X(x,z) \le \diam(T_x) \le (1/\epsilon + 20)2^i\eps < 2^i/\eps$. For (b) induction hypothesis and observe that for any point $z$ in $F_x$, we have $d_X(x, z) \le 2^i + (1/\epsilon + 20)2^i\eps$. From triangle inequality, we have $d(x',z) \ge d(x,x') - d_X(x, z) > 2^i/\eps + 2^i + (1/\epsilon + 20)2^i\eps > 2^i$.

\item By construction, either  (a) $T \in F_{i'}$, and in this case $\diam(T) \leq (1/\eps + 20)2^{i'} \leq (1/\eps + 20)2^{i}$ by induction, or (b) $T$ is formed by merging trees in $\{T_x,T_y\}\cup F_x \cup F_y$  where $x$ and $y$ are paired by $C_j$.  Recall that $\delta_X(x,y)\leq 2^i/\eps$, and that for every tree $A\in F_x\cup F_y$, there exists a leaf $z$ such that $\delta_X(x,z)\leq 2^i$ or $\delta_X(y,z)\leq 2^i$. Thus, by triangle inequality and induction, it follows that:
	\begin{equation*}
	\begin{split}
	\diam(T) &\leq  2^i/\eps + \diam(T_x) + 2\cdot 2^i + 2\max_{A\in F_x} \diam(A)  \\
	& \quad + \diam(T_y) + 2\cdot 2^i + 2\max_{B\in F_y} \diam(B) \\
	&\leq 2^i/\eps + 6(1/\eps + 20)2^{i'} + 4\cdot 2^i \\
	&\leq 2^i/\eps + 6(1/\eps + 20)2^{i}\eps + 4\cdot 2^i  \\
	&\leq (1/\eps + 10 + 120\eps)2^i \leq  (1/\eps + 20)2^i 
	\end{split}
	\end{equation*} 
	when $\eps \leq 1/12$, as desired.
\end{enumerate}
\end{proof}

We now show the robustness of the tree cover $\mathcal{T}$ assuming that $\eps \leq 1/12$. We will show that the stretch is $1+O(\eps)$; one can achieve stretch $1+\eps$ by scaling $\eps$.

 Let $x\not= y$ be any two points in $(X,\delta_X)$. Let $i$ be the non-negative integer such that:
\begin{equation}\label{eq:dxy}
2^{i-2}/\eps  < \delta_X(x,y)\leq 2^{i-1}/\eps~.
\end{equation}
Recall that we assume that the minimum distance in $X$ is larger than $1/(4\eps)$ and hence $i$ exists.   Let $p$ and $q$ be two net points of $N_i$ closest to $x$ and $y$, respectively. By the triangle inequality and \Cref{eq:dxy},  it holds that:
\begin{equation}\label{eq:dpq}
\begin{split}
\delta_X(p,q)&\leq \delta_X(x,y) + 2\cdot 2^i \leq  (1/2\eps  + 2) 2^i \leq 2^i/\eps ~  \quad \mbox{since }\eps\leq 1/12\\
\delta_X(p,q) &\geq  \delta_X(x,y) -  2\cdot 2^i >  (1/4\eps  - 2) 2^i > 0  \quad \mbox{since }\eps\leq 1/12
\end{split}
\end{equation}
 It follows from the second inequality in \Cref{eq:dpq} that  $p\not= q$. Since $\delta_X(p,q) \leq 2^i/\eps$, by property (2) of paring cover, there exists a set $C_j\in \mathcal{C}_i$ such that $p$ and $q$ are paired by $C_j$.  Let $T_p , T_q$ be the trees  in $F_{i'}$ and $F_p,F_q \subseteq F_{i'}$ associated with $p$ and $q$ as described in the construction. Let $T$ be the tree  in $F_i$ resulting from merging trees in $\{T_p, T_q\}\cup F_p \cup F_q$ by the algorithm. Since $\delta_X(x,p)\leq 2^i$, $x$ is a leaf of some tree  $T_x \in \{T_p\}\cup F_p$. By the same argument,  $y$ is a leaf of some tree   $T_y \in \{T_q\}\cup F_q$. Thus, both $x$ and $y$ are leaves in $T$. 
 
 Let $P$ be the path from $x$ to $y$ in $T$. Let $r, r_x,r_y$ be the roots of $T$, $T_x$,  and $T_y$, respectively.  Then $P$ consists of two paths $T_x[x,r_x]$, $T_y[y,r_y]$ and two edges $(r_x,r)$ and $(r_y,r)$. Let $Q$ be the path obtained from $P$ by replacing each internal vertex  $v$ of $P$ with a point chosen from a  leaf in $T_v$. We denote by $S(v)$ the leaf point chosen to replace each vertex $v\in P$. Let $Q_x$ (resp., $Q_y$) be the subpath of $Q$ from $x$ (resp., $y$) to $S(r_x)$ (resp., $S(r_y)$).  We have:
 
 \begin{equation}\label{eq:Q-weight}
 w(Q) \leq w(Q_x) + w(Q_y) + \delta_X(S(r_x), S(r)) + \delta_X(S(r), S(r_y))
 \end{equation}
 In the following claim, we bound the weight of each term in \Cref{eq:Q-weight}.
 
\begin{claim}\label{clm:Q-weight} $ \max\{w(Q_x),w(Q_y)\} \leq (2+40\eps) 2^i$ and $\delta_X(S(r_x), S(r) ) + \delta_X(S(r_y), S(r) ) \leq  \delta_X(x, y) + 4(5+60\eps)2^i $.
 \end{claim}
\begin{proof}[Proof of \Cref{clm:Q-weight}]
 Recall that $T_x,T_y$ are trees in $F_{i'}$, where $i' = i - \lceil\log(1/\eps) \rceil $. By \Cref{lm:diam-Fi}, we have:
	
	\begin{equation}\label{eq:Qx-Qy}
	\begin{split}
	w(Q_x) &= \sum_{j\leq i'} (1/\eps +20)2^{j} \leq (1/\eps + 20)2^{i'+1} \leq (1/\eps + 20)2\eps 2^i\leq  (2+40\eps) 2^i \\
	w(Q_y) &\leq (2+40\eps) 2^i  \qquad \mbox{(by the same argument)}
	\end{split}
	\end{equation}
	The summation for computing $Q_x$ in \Cref{eq:Qx-Qy} is due to the fact that for $T \in F_j$, we have $\diam(T) \le (1/\eps+20)2^j$.
	Let $Y$ and $Z$ be the sets of leaves of $\{T_p\}\cup F_p$ and $\{T_q\}\cup F_q$, respectively. Then we have:
	
	\begin{equation}\label{eq:diam-YZ}
	\begin{split}
	\diam(Y) &= \diam(T_p) + 2\cdot 2^i + 2\max_{A\in F_p} \diam(A)  \\
	&\leq  3(1/\eps + 20)2^{i'} + 2\cdot 2^i \leq 3(1/\eps + 20)\eps 2^{i} + 2\cdot 2^i = (5+60\eps)2^i \\
	\diam(Z) &\leq 	 (5+60\eps)2^i \qquad \mbox{(by the same argument)}
	\end{split}
	\end{equation}
	
	Observe that if $S(r) \in Y$, then $\delta_X(S(r), S(r_x)) \leq \diam(Y) \leq (5+60\eps)2^i $ by \Cref{eq:diam-YZ}. Otherwise, $S(r)\in Z$ and hence $ \delta_Y(S(r), S(r_x)) \leq (5+60\eps)2^i$. In either case, by the triangle inequality, we have:
	\begin{equation}\label{eq:dist-rx-ry-r}
	\begin{split}
	&\delta_X(S(r_x), S(r) ) + \delta_X(S(r_y), S(r) ) \\
	&\leq \delta_X(S(r_x), S(r_y)) + 2\min\{\delta_X(S(r_x), S(r)), \delta_X(S(r_y), S(r))\}\\
	&\leq \delta_X(S(r_x), S(r_y))  + 2(5+60\eps)2^i \\
	&\leq  \delta_X(x, y) + \diam(Y) + \diam(Z) + 2(5+60\eps)2^i  \qquad \mbox{(by the triangle inequality)}\\
	&\leq  \delta_X(x, y) + 4(5+60\eps)2^i   \qquad \mbox{(by \Cref{eq:diam-YZ})}
	\end{split}
	\end{equation}
	\end{proof}

	By \Cref{eq:Q-weight} and \Cref{clm:Q-weight}, we have that:
	\begin{equation}
	\begin{split}
	w(Q) &\leq  2(2+40\eps) 2^i  + \delta_X(x, y) + 4(5+60\eps)2^i    \\
	&\leq \delta_X(x, y)  + O(1)2^i \qquad \mbox{(since $\eps \leq 1/12$)}\\
	&\leq \delta_X(x, y) + O(1) 4\eps\delta_X(x,y) \qquad \mbox{(by \Cref{eq:dxy})}\\
	&= (1+O(\eps))\delta_X(x,y)~,
	\end{split}
	\end{equation}
	as claimed.

\subsection[Implementing Robust Tree Cover in \texorpdfstring{$O(n\log n)$}{O(nlogn)} time]{Implementing Robust Tree Cover in $\bm{O(n\log n)}$ time} \label{subsec:implement-RTC}

 To make our construction of a robust tree cover efficient, we need two data structures: 
\begin{itemize}
	\item[(a)] An implicit representation  of a  hierarchy of nets using $O_{\eps,d}(n)$ space. Let $\hat{N}_i \subseteq N_i$ be the subset of the net points that are explicitly stored in the hierarchy at level $i$. We have that $\sum_{i\geq 0} |\hat{N}_i| = O_{\eps,d}(n)$.
	\item[(b)]  For each net point $p \in \hat{N}_i$, store all the points in $N_i\cup N_{i'}$ within a distance $O(1/\eps)2^i$ from $p$, assuming that there is at least one such point other than $p$.   (If there are no such points, then $p$ will not be explicitedly store at level $i$ in the hierarchy.) Recall that $i' = i - \lceil \log(1/\eps)\rceil$ and hence the  pairwise distance of points stored at $p$ in (b) is at least $2^{i-\log\lceil 1/\eps\rceil} = \Omega(\eps 2^i)$. It follows from the packing lemma (\Cref{lm:packing-dd}) that the number of stored points is $\eps^{-O(d)}$.
\end{itemize}
Both data structures (a) and (b) can be constructed in $O_{\eps,d}(n\log{n})$ time using the result by Cole and  Gottlieb~\cite{CG06B}.  

Given these data structures, in Step 1, constructing partition $\mathcal{P}_i$ can be done in $O_{\eps,d}(|\hat{N}_i|)$ time following exactly the algorithm where $\hat{N}_i\subseteq N_i$ is the set of points stored explicitly in the hierarchy of nets. Thus, constructing  $C(\mathcal{P}_i)$ can also be done in $O_{\eps,d}(|\hat{N}_i|)$ time.  Note that $\sum_{i\geq 0}\hat{N}_i = O_{\eps,d}(n)$ by (a). It follows that the  total running time of Step 1 is $O_{\eps,d}(n)$. For Step 2, to construct the tree $T_{j,p}$, we need to identify for each point $x\in C_j$ the tree $T_x$ and forest $F_x$. For $T_x$, we can store a pointer to $T_x$ at $x$. For quickly identifying $F_x$, we (i) relax the definition of $F_x$ to contains subtrees of $F_{i'}$ such that each contains a \emph{net point} of $N_{i'}$ within distance $O(2^i)$ from $x$, (ii) guarantee that each tree in $F_{i'}$  contains at least one point in $N_{i'}$. Thus, we can identify $F_x$ in $O(\eps^{-d}) = O_{\eps,d}(1)$ time by looking at all points stored at $x$ in data structure (b). To guarantee (ii) inductively, in Step 2, we not only merge trees from pairs $x,y$ in $C_j$, but also merge trees in $F_{i'}$ that contain points close to points in $N_i$. Specifically, for each point $z\in \hat{N}_i$, which might not be in $C_j$, we merge the tree containing $z$ and all the (unmerged) trees in $F_{i'}$ containing points  of $N_{i'}$ within distance $2^i$ from $z$. This can be done in  $O_{\eps,d}(|\hat{N}_i|)$ time. Thus, the total running time of both steps is $O_{\eps,d}(n)$, and the final running time is  $O_{\eps,d}(n\log(n))$.

 \subsection{Deriving a fault-tolerant navigation (and routing) scheme}\label{sec:ft-navigation}
In the navigation scheme presented in \Cref{SectionPathOracle}, we did not exploit a crucial property of the tree cover theorem in doubling metrics \cite{BFN19B}: For every pair $u,v$ of points in $M_X$, there is a $(1+\eps)$-spanner path in one of the trees in the cover --- such that the path starts and ends at leaves corresponding to $u$ and $v$. To achieve FT navigation algorithm, we must rely on this property.
For any two points from a doubling metric, the navigation algorithm from \Cref{SectionPathOracle} locates points, which are now the leaves in the corresponding tree of the tree cover. Then, it uses the navigation scheme for that particular tree to navigate between these points. Every vertex in the tree is associated with a single point in the metric space, hence while navigating the tree we can directly obtain the information about the path in the metric space.
In the case of FT navigation, every vertex in the tree stores (or is associated with) $f+1$ points (rather than one) that correspond to its descendant leaves. This is the case for all the vertices, except for ones with less than $f+1$ descendant leaves (including the leaves themselves); such vertices store all their descendant leaves.
To navigate between any two non-faulty points $u$ and $v$ (corresponding to leaves in the tree), we apply the same navigation scheme as given in \Cref{SectionPathOracle}, but for every vertex that we traverse along the path in the tree, we pick a non-faulty point stored in that vertex {\em arbitrarily}, if it stores $f+1$ points. For every vertex with less than $f+1$ leaves in its subtree,  
it must store either $u$ or $v$, since all the nodes along the path in the tree are ancestors of either $u$ or $v$. Since both $u$ and $v$ are non-faulty, we will have a non-faulty point to choose from ($u$ or $v$ or both). The query time of the navigation scheme remains $O(k)$. The basic (non-FT) routing scheme is deferred to \Cref{sec:labeling}; however, 
equipped with the FT-navigation scheme that we've just described, it is straightforward to strengthen the basic routing scheme to achieve fault-tolerance (with the size bounds growing by a factor of $f$).
\section{Applications} \label{sec:applications}

We argue that an efficient navigation scheme is of broad potential applicability, by providing several applications.
The first and perhaps the most important application is an efficient compact routing scheme, which is given in \Cref{sec:labeling}. Next, in \Cref{sec:sparse-spanner} we show that our navigation technique can be used for efficient spanner sparsification without increasing its stretch and lightness by much. In \Cref{sec:spt,sec:mst} we show that one can use the navigation technique to compute the SPT and MST on top of the underlying spanner. Finally, in \Cref{sec:treeProduct}, we address two related problems: online tree product and MST verification.

\subsection{Compact routing schemes} \label{sec:labeling}

A routing scheme is a distributed algorithm for delivering packets of information from any given source node to any specified destination in a given network. Every node has a \emph{routing table}, which stores local routing-related information. In addition, every node is assigned a unique \emph{label} (sometimes called \emph{address}).
Given a destination node $v$, routing algorithm is initiated at source $u$ and is given the label of $v$. Based on the local routing table of $u$ and the label of $v$, it has to decide on the next node $w$ to which the packet should be transmitted. More specifically, it has to output the relevant \emph{port number} leading to its neighbor $w$. 
Each packet of information has a \emph{header} attached to it, which contains the label of the destination node $v$, but may contain additional information that may assist the routing algorithm. Upon receiving the packet, any intermediate node $w$ has at its disposal its local routing table and the information stored in the header. This process continues until the packet arrives at its destination (node $v$). 

We consider routing in metric spaces, where each among $n$ points in the metric corresponds to a network node. Initially, we choose a set of links that induces an overlay network over which the routing must be performed. We would like the overlay network size to be small and yet to be able to route using very few hops. The challenge is to do so while also optimizing the tradeoff between the storage (i.e., size of routing tables, labels, and headers) and the stretch (i.e., the ratio between the distance packet traveled in the network and the distance in the original metric space).
In addition, one may try to further optimize the time it takes for every node to determine (or output) the next port number along the path, henceforth \emph{decision time}, and other quality measures, such as the maximum degree in the overlay network.

There are two models, based on the way labels are chosen: \emph{labeled}, where the designer is allowed to choose (typically $\polylog(n)$) labels, and \emph{name-independent}, where an adversary chooses labels. Depending on the way the port numbers are assigned, we distinguish between the \emph{designer-port} model, where the designer can choose the port number, and the \emph{fixed-port} model, where the port numbers are chosen by an adversary. 
For an additional background on compact routing schemes, we refer the reader to~\cite{Peleg00, TZ01, FG01, AGGM06, Che13}.

Our basic result is an efficient routing scheme for tree metrics, which we present in \Cref{subsec:treeRouting}. Our routing scheme works in the labeled, fixed-port model. 
Next, in \Cref{subsec:routingCovers}, we apply the routing scheme for tree metrics on top of the collection of trees provided by any of the aforementioned tree cover theorems (cf. \Cref{tab:treeCover}) and obtain routing schemes for various metrics. This application is nontrivial for general graphs, as it aims at optimizing the label sizes.

\subsubsection{Routing scheme for tree metrics}\label{subsec:treeRouting}

We show that one can construct an efficient 2-hop routing scheme for tree metrics. The guarantees are summarized in the following theorem.

\begin{theorem}\label{thm:treeRouting}
Let $M_T$ be a tree metric represented by an edge-weighted tree $T$ with $n$ vertices. We can preprocess $T$ it in $O(n \log{n})$ time and construct a routing scheme which works on the overlay network with $O(n\log{n})$ edges. The routing scheme works in the labeled, fixed-port model, uses routing tables and labels of $O(\log^2{n})$ bits and headers of $\lceil\log{n}\rceil$ bits and routes in at most 2 hops and in $O(1)$ decision time. 
\end{theorem}

The first step of the preprocessing phase consists of constructing spanner $G_T$ of $T$ and the navigation data structure $\mathcal{D}_T$ as described in \Cref{thm:treeNavigate}. The number of edges of $G_T$ is $O(n\alpha_k(n))=O(n\log{n})$ for $k=2$. In addition to $G_T$, we obtain an augmented recursion tree $\Phi$ for navigating $G_T$. By \Cref{lem:depth}, the depth of $\Phi$ is $O(\log{n})$. (More details on this construction can be found in \Cref{subsec:treespanner}.) We will show how to construct a routing scheme for spanner $G_T$ using the augmented recursion tree $\Phi$.
Recall that, when $k = 2$, every non-leaf vertex in $\Phi$ corresponds to exactly one vertex in $T$.

We shall assume that each vertex in $T$ is assigned a unique identifier between 1 and $n$, which can be justified via a straightforward linear time procedure. At the beginning, we preprocess the augmented recursion tree $\Phi$ using the lowest common ancestor (LCA) labeling scheme by \cite{AHL14}.\footnote{This scheme is usually known as \emph{nearest common ancestor (NCA)} scheme, but we use term lowest common ancestor (LCA) in order to be consistent with the rest of the paper.} This scheme uses linear preprocessing time to assign $O(\log{n})$-bit label to each vertex in the tree so that any subsequent LCA query can be answered in constant time. 
Recall from \Cref{subsec:treespanner} that every vertex $u$ in $T$ uniquely corresponds to a vertex in $\Phi$, denoted by $u.ptr(\Phi)$.
For every $u$ in $T$, let $lca(u)$ be the LCA label of $u.ptr(\Phi)$.

We start by describing the label of each vertex in $T$.
Fix a vertex $u \in T$ and denote by $\alpha$ its corresponding vertex in $\Phi$, i.e., $\alpha \coloneqq u.ptr(\Phi)$.  Recall from \Cref{subsec:treespanner} that we use $\alpha.level$ to denote the level of $\alpha$ in $\Phi$. Let $\beta_j$ be the ancestor of $\alpha$ at level $j$ in $\Phi$, so that $\beta_0 = rt(\Phi)$ and $\beta_{\alpha.level} = \alpha$. 
Let $v_j \coloneqq \beta_j.ptr(T)$ be the vertex in $T$ which corresponds to $\beta_j$ for all $0 \le j \le \alpha.level-1$.
Denote by $h_u$ a 2-level hash table with key $lca(v_j)$ and value $\port(v_j, u)$ for each $0 \le j \le \alpha.level-1$.
This table occupies space linear in the size of all its key-value pairs. Since there is at most $O(\log{n})$ ancestors (the depth of $\Phi$ is $O(\log{n})$ by \Cref{lem:depth}) and each ancestor takes $O(\log{n})$ bits of space to store, the total memory required for $h_u$ is $O(\log^2{n})$. Any subsequent query for an element in $h_u$ takes worst-case constant time. Details of implementation can be found in \cite{FKS84,AN96}.
Label of $u$, denoted by $\lab(u)$, consists of $\lab(u) = (lca(u), h_u)$. Its size is $O(\log^2(n))$ bits.

The routing table of node $u$, denoted by $\tab(u)$ contains information similar to its label. Denote by $h'_u$ a 2-level hash table with key $lca(v_j)$ and value $\port(u, v_j)$ for each $0 \le j \le \alpha.level-1$.
Note that in $h_u$, we store port numbers leading to $u$ from its ancestors (with respect to $\Phi$), whereas in $h'_u$ we store port number from $u$ to its ancestors. 
In addition, if $u$ corresponds to a leaf in the augmented recursion tree, we add to $\tab(u)$ an array $\bas(u)$ containing a constant number of pairs $(v, \port(u, v))$ for every other node $v$ corresponding to the same leaf. Since $u$ might not be directly connected to $v$, $\port(u,v)$ denotes the port number leading to the first vertex on the shortest path from $u$ to $v$.
Recall that, by \Cref{property:base-case}, there can be at most $O(k) = O(1)$ vertices corresponding to the same base case, meaning that $\bas(u)$ contains a constant number of $O(\log{n})$-bit entries. The routing table of $u$ consists of $\tab(u) = (lca(u), h'_u, \bas(u))$, which takes $O(\log^2{n})$ bits to store.
We note that labels and routing tables can be computed in $O(n\log{n})$ time using the augmented recursion tree and techniques similar to those in \Cref{SectionPathOracle}.

We now specify the routing protocol. Recall that, upon a query to route from $u$ to $v$, the algorithm is executed on node $u$, has access to the local routing table of $u$ (denoted by $\tab(u)$), and is passed the label of the destination $v$ (denoted by $\lab(v)$). 
We shall assume that $u\neq v$, since otherwise there is nothing to do.
First, we check if $u$ and $v$ correspond to the same base case in $\Phi$. We do so by looking for $v$ in $\bas(u)$, if $\bas(u)$ is nonempty. This can be done in constant time since the number of entries in $\bas(u)$ is constant. If an entry corresponding to $v$ has been found, we extract from $\bas(u)$ the information about the port corresponding to the first edge on the shortest path from $u$ to $v$ and forward the packed with an empty header. By the guarantees of $G_T$, there exist a path of at most 2 hops between $u$ and $v$ in the subgraph induced by the vertices corresponding to the same base case as $u$ and $v$. Hence, the packet is either forwarded directly to $v$, in which case the algorithm successfully terminates, or it is forwarded to an intermediate node $w$ which corresponds to the same base case and has a direct link to $v$. In the latter case, the corresponding port can be extracted from $\bas(w)$, using the local routing table at node $w$.

If $u$ and $v$ do not correspond to the same base case, we look for the lowest common ancestor, $\lambda$, of $u.ptr(\Phi)$ and $v.ptr(\Phi)$. Let $lca(\lambda)$ denote the LCA label of this vertex. 
At this stage, we distinguish between several cases. If $u.ptr(\Phi) = \lambda$ (which can be checked using their LCA labels), then $u.ptr(\Phi)$ is an ancestor of $v.ptr(\Phi)$ in $\Phi$ and the underlying spanner $G_T$ contains an edge between $u$ and $v$. The corresponding port can be found in $h_v$, which is in $\lab(v)$. If $v.ptr(\Phi) = \lambda$, then $\lambda$ is an ancestor of $u.ptr(\Phi)$ and there is an edge between $u$ and $v$. The corresponding port can be found in $h'_u$, which is in $\tab(u)$. If none of the above is the case (recall that we assumed that $u\neq v$), by the design of $G_T$, there is edge between $u$ and the vertex $w$ corresponding to $\lambda$ and from $w$ to $v$. From $\lab(v)$, we extract $\port(w, v)$, store it in the header and forward the packet to $\port(u, w)$, which can be found in $\tab(u)$. This completes the description of the routing algorithm.

\subsubsection{Routing in metric spaces}\label{subsec:routingCovers}
We proceed to show how to employ the described routing scheme for metric spaces, thus proving \Cref{thm:routing}. 

Similarly to what has been done in \Cref{SectionPathOracle}, for a given metric $M_X=(X,\delta_X)$ we first construct one of the tree covers from \Cref{tab:treeCover}. Denote the stretch of this cover by $\gamma$ and the number of the trees by $\zeta$. 
The underlying graph $H$ is the union of the trees in the cover --- it has the same vertex set as $M_X$, denoted by $X$, and has an edge set obtained as a union of the edges of $\zeta$ trees in the cover.
At this stage, port numbers are assigned (by an adversary) for every vertex $v$ in $\{1,\dots, \deg_H(v)\}$, where $\deg_H(v)$ is the degree of $v$ in $H$. Then, for each tree in the cover, we construct a $2$-hop routing scheme as provided by \Cref{thm:treeRouting}. 
We distinguish between cases where we use tree covers and Ramsey tree covers.

\paragraph{Routing using tree covers.}
In addition to the routing schemes, we employ the \textit{distance labeling scheme} with a stretch of $(1+\eps)$, given by \cite{FreedmanLabelingScheme}. Their labeling scheme uses $O(\log (1/\eps)\log n)$ bits of space per vertex and achieves a constant query time. In our routing scheme, each node stores $\zeta$ distance labels, one per tree in the cover, both as a part of its routing table and as a part of its label. In addition, for each of the trees in the cover, we store label and routing table, as described in \Cref{subsec:treeRouting}.
This means that the memory consumed by the routing table at each node is $O(\zeta \cdot \log^2{n} + \zeta \cdot \log (1/\eps)\log n) = O(\zeta\cdot \log{n}\cdot \log(n/\eps))$.

The routing algorithm executed at node $u$ first queries the distance labeling scheme for the approximate distance between $u$ and the destination node $v$ in each of the trees. Since each query takes constant time, this step requires time proportional to $\zeta$. 
The routing proceeds using the routing information of the tree which has the smallest stretch. Suppose that such a tree had index $i$. We find the next port using the $i$th entry of the routing table of $u$ and the label of $v$. If the next step in tree routing correspond to the base case, we only transfer the index $i$ in the header. From the algorithm in \Cref{subsec:treeRouting}, we know that the next step is either $v$, or we route via an intermediate node $w$ which has port number $\port(w,v)$ in its routing table corresponding to the $i$th tree. If the next step does not correspond to the base case, we either route directly, in which case there is no need to store anything in the header, or route via an intermediate node $w$. In the latter case, we can extract the information on $\port(u, w)$ and $\port(w, v)$ from the entries in $\tab(u)$ and $\lab(v)$ (corresponding to the $i$th tree) and route to $\port(u,w)$, while sending only $\port(w,v)$ in the header. In both cases, the header size is $\lceil{\log{n}}\rceil$.
Finally, notice that we are using the labeling schemes which return $(1+\eps)$-approximate distances, meaning that this step incurs an additional $(1+\eps)$ factor to the stretch of our routing path. The stretch of $(1+\eps)$ can be achieved by appropriate scaling.

\paragraph{Routing using Ramsey tree covers.}
When using Ramsey tree covers, we know for each node which of the trees in the cover achieves the desired stretch. The label of a node is now comprised of its label in the routing table for that particular tree, together with the index of that tree. Routing table of every node contains routing tables for each of the $\zeta$ trees in the cover. Hence, the label sizes of this scheme are $O(\log^2{n})$, which is the size of routing schemes for tree spanners in \Cref{thm:treeRouting}, and the routing tables have size $O(\zeta \cdot \log^2{n})$.
Given a source node $u$ and the label of destination $v$, the algorithm uses routing table corresponding to the tree which index is in the label of $v$. In other words, we can in constant time decide which among the $\zeta$ routing tables stored at $u$ to use.

In conclusion, we have proved \Cref{thm:routing}, whose guarantees are summarized in the \Cref{tab:routing}.

\begin{table}[H]
\centering
\begin{tabular}{|c|c|c|c|c|}
\hline
\multicolumn{1}{|c|}{\multirow{2}{*}{\textbf{stretch}}} & \multicolumn{2}{c|}{\textbf{storage}}      & \multicolumn{1}{c|}{\multirow{2}{*}{\textbf{time}}} & \multicolumn{1}{c|}{\multirow{2}{*}{\textbf{metric}}} \\ \cline{2-3}
\multicolumn{1}{|c|}{}                                  & \multicolumn{1}{c|}{\textbf{table}}   & \multicolumn{1}{c|}{\textbf{label}} & \multicolumn{1}{c|}{}                               & \multicolumn{1}{c|}{}                                 \\ \hline
1 & \multicolumn{2}{|c|}{$O(\log^2{n})$} &  $O(1)$ & tree \\\hline
$1+\eps$& \multicolumn{2}{|c|}{$O(\eps^{-O(d)}\log(n)  \log(n/\eps) )$} &  $O(\eps^{-O(d)})$ & doubling dim. $d$\\\hline
$1+\eps$ & \multicolumn{2}{|c|}{$O((\log{n}/\eps)^3 \log(n) )$}  & $O((\log{n}/\eps)^2 )$ & fixed-minor-free \\\hline
$O(\ell)$ & $O(\log^2{n})$ & $O(\ell  n^{1/\ell}  \log^2{n})$   &$O(1)$& general \\\hline
$O(n^{1/\ell} \log^{1-1/\ell}{n})$ &$O(\log^2{n})$ & $O(\ell \cdot \log^2{n})$   &$O(1)$& general\\\hline
\end{tabular}
\caption{Summary of our results for 2-hop routing schemes. For each result, header size is $\lceil{\log{n}}\rceil$. In the last two results, parameter $\ell\ge 1$ is an arbitrary integer. }\label{tab:routing}
\end{table}

\subsection{Fault-tolerant routing}\label{sec:ft-routing}
\paragraph{Fault-tolerant routing in trees.}
We describe a fault-tolerant routing scheme for trees of the robust tree cover presented in \Cref{ft}. 
Given a tree $T$ from the robust tree cover, we first compute a spanner for it, as described in \Cref{sec:ft-spanner}. Recall that for every $v \in T$, we choose a set $R(v)$ corresponding to (at most) $f+1$ leaves of the subtree of $T$ rooted at $v$. We show how to route between any two leaves of this tree, by slightly modifying the tree routing scheme presented in \Cref{subsec:treeRouting}.

For every leaf in $u \in T$ (corresponding to some point in the metric), let $\alpha$ be its corresponding vertex in $\Phi$, i.e., $\alpha = u.ptr(\Phi)$. We construct a 2-level hash table $h_u$ for $u$ as follows. For the $j$th ancestor of $\alpha$ in $\Phi$, let $v_j$ be its corresponding point in $T$ and let $w_1, \ldots, w_\ell$ be the leaves in $T$ from $R(v_j)$, ordered increasingly by their identifiers. We add to the hash table $h_u$ an entry with key $lca(v_j)$ and value $\lAngle{\port(w_1, u),\ldots\port(w_\ell, u)}$. The label of $u$ is $\lab(u) = (lca(u), h_u)$; it requires $O(f\log^2{n})$ bits of space.

We next describe the routing table of $u$, denoted by $\tab(u)$. Let $h'_u$ be a 2-level hash table, where for the $j$th ancestor of $\alpha$, we store key $lca(v_j)$ and value $\lAngle{\port(u, w_1),\ldots\port(u, w_\ell)}$. If $u$ corresponds to a base case of the spanner construction, we keep a routing table $\bas(u)$ for every other node $v$ corresponding to the same leaf in $\Phi$. In particular, let $w_1, \ldots, w_\ell$ be the leaves in $T$ in $R(v)$. We add an entry $(v, \lAngle{\port(u, w_1), \ldots, \port(u, w_\ell)})$ to $\bas(u)$. The routing table of $u$, denoted by $\tab(u)$ consists of $(lca(u), h'_u, \bas(u))$; it requires $O(f\log^2{n})$ bits of space.

Given two leaves $u$ and $v$ in $T$, the routing algorithm is similar to the one from \Cref{subsec:treeRouting}. The only difference is that when we route via an intermediate vertex $w \in T$, we scan entries corresponding to $R(w)$ to find a non-faulty vertex. Since the entries in $R(w)$ are ordered increasingly by their identifiers, we can find a port to a non-faulty vertex in $R(w)$ in $O(f)$ steps. This concludes the description of the fault-tolerant routing scheme for trees.

\paragraph{Fault-tolerant routing in metric spaces.}
Suppose we are given a metric $M_X=(X, \delta_X)$ with doubling dimension $d$. First, we construct a fault-tolerant spanner for $M_X$, as described in \Cref{sec:ft-spanner}, and assign port numbers according to spanner edges. The rest of the construction remains the same as for the routing using tree covers described in \Cref{subsec:routingCovers}. For each of the $\zeta = O(\eps^{-O(d)})$ trees in the cover, we keep labels and routing tables of size $O(f\log^2{n})$. In addition, for each tree, we construct a distance labeling scheme given by \cite{FreedmanLabelingScheme}. Labels for this scheme occupy $O(\log(1/\eps)\log{n})$ bits of space. Hence, the total memory per vertex in the tree is $O(\zeta\cdot\log{n}\cdot(\log(1/\eps) + f\log{n})) = O(\eps^{-O(d)}\cdot\log{n}\cdot(\log(1/\eps) + f\log{n}))$. The routing algorithm remains the same. This concludes the description of the fault-tolerant routing; the result is summarized in the following theorem.

\begin{theorem}
For any $n$-point metric $M_X = (X, \delta_X)$ with doubling dimension $d$, one can construct a $(1+\eps)$-stretch $2$-hop routing scheme in the labeled, fixed-port model with headers of $\lceil\log{n}\rceil$ bits, labels and local routing tables of $O(\eps^{-O(d)}\cdot\log{n}\cdot(\log(1/\eps) + f\log{n}))$ bits, and local decision time $O(f)$.
\end{theorem}

\subsection{Spanner sparsification}\label{sec:sparse-spanner}
Let $M_X = (X,\delta_X)$ be an arbitrary $n$-point metric space and let $G_X$ be any \emph{light} $m$-edge spanner for $M_X$. Let $k\ge 2$ be an integer and let $H_k$ be a $k$-hop spanner for $M_X$ and $\mathcal{D}_x$ data structure for $H_k$ provided by~\Cref{thm:oracle}.
Our goal is to transform $G$ into a sparse spanner for $M_X$, without increasing the stretch and weight by much.

The transformation is as follows. For each edge $(u, v)$ in $E(G_X)$, we query the data structure $\mathcal{D}_X$ for a $k$-hop path $P_{u, v}$ between $u$ and $v$ in $H_k$; the output is $G'_X \coloneqq (X, \cup_{(u, v) \in E(G_X)}P_{u, v})$. The following theorem summarizes the guarantees of this sparsification procedure. See also \Cref{tab:sparsify}.

\begin{theorem}\label{thm:sparsification}
Let $G_X$ be a $k$-hop spanner for an $n$-point metric space $M_X=(X, \delta_X)$, with size $m$, stretch $\beta$, and lightness $L$; let $H_k$ be a spanner for $M_X$ equipped with a data structure $\mathcal{D}_X$ as in \Cref{thm:oracle}. Then, one can in $O(m\cdot \tau)$ time transform $G_X$ into spanner $G'_X$ with stretch $\gamma\cdot \beta$, lightness $\gamma \cdot L$ and $O(n \alpha_k(n) \cdot \zeta)$ edges.
\end{theorem}

\begin{proof}
For each edge of $G_X$ we perform a query as in \Cref{thm:oracle}, and each query takes $O(\tau)$ time, so the total running time of the transformation is $O(|E(G_X)|\cdot \tau)$.
The stretch of a spanner is equal to the maximum stretch between any two adjacent points in it; hence, $stretch_{G'}(u, v) \le \gamma \cdot stretch_{G}(u, v) \le \gamma \cdot \beta$, where the first inequality follows since each edge $(u, v)$ in $G_X$ got replaced by a path of total weight $\gamma \cdot w(u, v)$, and the second inequality follows since the stretch of $G$ is $\beta$.
Similarly, since each edge $(u,v)$ of $G_X$ got replaced by a path $P_{u,v}$ of weight at most $\gamma \cdot w(u,v)$, the lightness of the resulting spanner is at most $\gamma \cdot L$.
The resulting spanner is a subgraph of $H_k$ so it has at most $O(n\alpha_k(n)\cdot \zeta)$ edges.
\end{proof}

\begin{table}[H]
\centering
\scalebox{0.94}{
\begin{tabular}{ |c | c | c | c | c | }
\hline
\textbf{Stretch} & \textbf{Size} & \textbf{Lightness} & \textbf{Time} & \textbf{Metric family} \\\hline
$(1+\eps)\beta$ & $O(n\alpha_k(n) /\eps^{O(d)})$ & $(1+\eps)L$ &$O(mk/\eps^{O(d)})$ & doubling dim. $d$ \\\hline
$(1+\eps)\beta$ & $O(n \alpha_k(n)\eps^{-2} \log^2{n})$ & $(1+\eps)L$ & $O(mk\eps^{-2}\log^2{n})$ & fixed-minor-free  \\\hline
$O(\ell\cdot \beta$) & $O(n \alpha_k(n)\ell n^{1/\ell})$ & $O(\ell\cdot L)$ & $O(m\cdot k)$ & general \\\hline
$O(\beta n^{1/\ell} \log^{1-1/\ell}{n} )$ & $O(n \alpha_k(n) \cdot \ell)$ & $O(L n^{1/\ell} \log^{1-1/\ell}{n} )$ & $O(m\cdot k)$ & general \\\hline
\end{tabular}
}
\caption{Summary of our result for sparsification of spanner $G_X$ for an $n$-point metric space. We use $\beta$ to denote the stretch of $G_X$, which $m$ edges and lightness $L$. In the last two entries, $\ell \ge 1$ denotes an arbitrary integer.}\label{tab:sparsify}
\end{table}
\begin{remark}
For fixed $\eps$ and doubling dimension $d$, \Cref{thm:sparsification} gives rise to a transformation that works in $O(mk)$ time and produces a spanner of size $O(n\alpha_k(n))$ with stretch and lightness increased by a factor of $(1+\eps)$. Similarly, for fixed $\eps$ and fixed-minor-free metrics, we obtain a transformation that works in $O(mk\log^2{n})$ time and produces a spanner of size $O(n\alpha_k(n)\log^2(n))$ with stretch and lightness increased by a factor of $(1+\eps)$.
\end{remark}

\subsection{Approximate shortest path trees}\label{sec:spt}
Once a spanner has been constructed, it usually serves as a ``proxy'' overlay network, on which any subsequent computation can proceed, in order to obtain savings in various measures of space and running time.
Thus, we shall focus on devising efficient algorithms that run on the spanner itself.
In some applications, we may not have direct access to the entire spanner, but may rather have implicit and/or local access, for example by means of a data structure for approximate shortest paths within the spanner, such as the one provided by~\Cref{thm:oracle}.

In this section, we explain how our navigation technique can be used for efficiently computing an approximate shortest-path tree (SPT) that is a subgraph of the underlying spanner. 
In any metric, its SPT is simply a star, which is most likely not a subgraph of the underlying spanner. 
Assuming that we have direct, explicit access to the spanner, we can simply compute an SPT on top of it using Dijkstra's algorithm, which will provide an approximate SPT for the original metric. 
For an $n$-vertex spanner, this approach will require $\Omega(n \log n)$ time, even if the spanner size is $o(n \log n)$. 
There is also an SPT algorithm that runs in time linear in the spanner size, but it is more complex and also assumes that $\log n$-bit integers can be multiplied in constant time \cite{DBLP:journals/jacm/Thorup99}. 
Using our navigation scheme, as provided by~\Cref{thm:oracle}, we can do both better and simpler, and we don't even need explicit access to the underlying spanner (though we do need, of course, access to the navigation scheme).  
The data structure provided by~\Cref{thm:oracle} allows us to construct, within time $O(n \tau)$, an
approximate SPT. In particular, for low-dimensional Euclidean and doubling metrics, 
we can construct a $(1+\eps)$-approximate SPT (for a fixed $\eps$) that is a subgraph of the underlying spanner within $O(n k)$ time, where $k = 2,3,\ldots,O(\alpha(n))$. 

In what follows, we will assume that we are given a metric $(X, \delta_X)$, with $|X| \coloneqq n$ and that we have constructed spanner $H_k$ and the data structure $\mathcal{D}_X$ for $(X, \delta_X)$ as in \Cref{thm:oracle}. Recall that we use $\gamma$ to denote the stretch of the path returned by $\mathcal{D}_X$ and $\tau$ denotes the time spent per query.
The algorithm for computing an approximate SPT rooted at given vertex $rt$ is stated in procedure $\textsc{ApproximateSPT}(rt)$.

\begin{algorithm}
\begin{algorithmic}[1]
\Procedure{ApproximateSPT}{$rt$}
\ForInline {$v \in V$}{let $\pi(v) \gets \emptyset$}
\ForInline {$v \in V \setminus \{rt\}$}{let $\dst(v) \gets \infty$}
\State Let $\dst(rt) \gets 0$
\State Let $V(T) \gets \{rt\}$ and $E(T) \gets \emptyset$
\For {$v \in V \setminus \{rt\}$}\label{line:SPT:loop}
\State Query for the $k$-hop $\gamma$-approximate shortest path $P_{rt, v}$ from $rt$ to $v$;
\State For each edge $e = (x,y)$, ordered from $rt$ to $v$ along $P_{rt, v}$, invoke \Call{Relax}{$x,y$}\label{line:SPT:invokeRelax}
\EndFor
\EndProcedure
\Procedure{Relax}{$u,v$}
\If {$\dst(u) = \infty$}
\State $V(T) \gets V(T) \cup \{v\}$
\State $E(T) \gets E(T) \cup \{(u,v) \}$
\State $\dst(v) \gets \dst(u) + w(u,v)$\label{line:SPT:changeDist1}
\State $\pi(v) \gets u$\label{line:SPT:changePi1}
\ElsIf {$\dst(u) + w(u,v) < dst(v)$}
\State $E(T) \gets (E(T) \setminus \{(\pi(v), v) \}) \cup \{(u,v) \}$
\State $\dst(v) \gets \dst(u) + w(u,v)$\label{line:SPT:changeDist2}
\State $\pi(v) \gets u$\label{line:SPT:changePi2}
\EndIf
\EndProcedure
\end{algorithmic}
\caption{Computing approximate shortest path trees.}\label{alg:approxSPT}
\end{algorithm}
We shall prove three claims which will imply the guarantees of the algorithm in \Cref{alg:approxSPT}. Intuitively, \Cref{claim:tree} proves that, throughout the execution, the graph $T$ maintained by the algorithm is a tree. Then, \Cref{claim:dist} shows that the value $\dst$ for each vertex $v\in T$ will be an upper bound on its distance from the root, denoted by $\delta_T(rt, v)$. Finally, \Cref{claim:stretch} implies that upon termination, all the vertices have a $\gamma$-stretch path to the root.

\begin{claim}\label{claim:tree}
Throughout the execution of $\textsc{ApproximateSPT}(rt)$, graph $T$ is a tree.
\end{claim}
\begin{proof}
Initially, the claim holds vacuously, since $T$ consists of vertex $rt$ and no edges.
The structure of $T$ changes in \cref{line:SPT:changePi1,line:SPT:changePi2}. It is important to notice that the order of relaxations performed at \cref{line:SPT:invokeRelax} allows us to assume that whenever $\textsc{Relax}(u, v)$ is executed, $u$ is in $T$.  We will assume $T$ is a tree before the relaxation and prove that it will also be the case after the relaxation. If $v$ was in $T$ earlier, then $T$ changes a parent of $v$, thus remaining connected while preserving the number of vertices and edges. If $v$ was not in $T$, then $T$ grows by one vertex and one edge. In both cases, $T$ is a connected graph with $|V(T)|$ vertices and $|V(T)|-1$ edges, so it is a tree.
\end{proof}

\begin{claim}\label{claim:dist}
For each vertex $v \in T$ and for each ancestor $u$ of $v$, the algorithm $\textsc{ApproximateSPT}(rt)$ maintains invariant $\dst(u) + \delta_T(u, v) \le \dst(v)$.
\end{claim}
\begin{proof}
At the beginning of the execution, the invariant is trivially satisfied.
The values $\dst$ change in \cref{line:SPT:changeDist1,line:SPT:changeDist2}.  Similarly, the values $\delta_T$ change in \cref{line:SPT:changePi1,line:SPT:changePi2}, due to the change of edges in $T$. 
We assume that the invariant was true before relaxation along $(u,v)$ and would like to show that it holds after. In particular, it suffices to show that (\emph{i}) it holds for $v$ and all its ancestors, (\emph{ii}) it holds for $v$ and all its descendants, and (\emph{iii}) it holds for any ancestor $x$ of $v$ and any descendant $w$ of $v$. %
As for (\emph{i}), the invariant trivially holds for $v$ as its own ancestor. Also, after the relaxation, we know that $\dst(v) = \dst(u) + w(u,v) = \dst(u) + \delta_T(u, v)$ so the claim holds for $u$ as a parent of $v$. The invariant was true before the relaxation, so for any ancestor $x$ of $u$ (and thus ancestor of $v$), $\dst(x) + \delta_T(x, u) \le \dst(u)$. Adding $\delta_T(u, v)$ to both sides of inequality completes this case. 
For (\emph{ii}), the invariant implies that $\dst'(v) + \delta_T(v, w) \le \dst(w)$, where $\dst'(v)$ was the value before the relaxation. After the relaxation, $\dst(v)$ decreased, so the invariant remains true. Finally, (\emph{iii}) holds since  $\dst(x) + \delta_T(x, w) = \dst(x) + \delta_T(x, v) + \delta_T(v, w) \le \dst(v) + \delta_T(v, w) \le \dst(w)$.
\end{proof}

\begin{claim}\label{claim:stretch}
Upon termination of $\textsc{ApproximateSPT}(rt)$, for any $v \in V$, $\delta_T(rt, v) \le \gamma\cdot \delta_X(rt, v)$.
\end{claim} 
\begin{proof}
We know that after the iteration of line \ref{line:SPT:loop} for vertex $v$, $\dst(v)$ is within $\gamma$ factor of $\delta_X(rt, v)$. In the subsequent iterations, $\dst(v)$ might only decrease due to relaxations. Finally, by \Cref{claim:dist} we know that $\delta_T(rt, v) \le \dst(v)$ throughout the execution of the algorithm.
\end{proof}

\Cref{claim:tree,claim:stretch} imply that $T$ is a $\gamma$-SPT for given metric with root at $rt$. The running time of the algorithm is dominated by $n-1$ queries of the path oracle, each of which takes $O(\tau)$ time as in \Cref{thm:oracle}. Thus, we have proved the following theorem.

\begin{theorem}\label{thm:spt}
Given a data structure $\mathcal{D}_X$ for an $n$-point metric space $M_X=(X, \delta_X)$ as in \Cref{thm:oracle}, one can construct a $\gamma$-approximate shortest path tree for $M_X$ rooted at any point $rt \in X$ in time $O(n \tau)$.
\end{theorem}

\begin{remark}
When $M_X$ is a metric of doubling dimension $d$, \Cref{thm:spt} gives rise to a construction of $(1+\eps)$-approximate SPTs in time $O(nk/\eps^d)$ time. If $M_X$ is a fixed-minor-free metric, the result is a $(1+\eps)$-SPT in time $O(nk\eps^{-2}\log^2{n})$. For general metrics and an integer parameter $\ell\ge 0$, one can construct an $O(\ell)$-SPT in time $O(nk)$.
\end{remark}

\subsection{Approximate Euclidean minimum spanning trees}\label{sec:mst}
Suppose that we would like to construct an approximate minimum spanning tree (MST). Here too, we shall focus on finding an approximate MST that is a subgraph of the underlying spanner.
In low-dimensional Euclidean spaces one can compute a $(1+\eps)$-approximate MST (for a fixed $\eps$) in $O(n)$ time \cite{chan2008well}, 
but again this approximate MST may not be a subgraph of the spanner. Running an MST algorithm on top of the spanner would require time that is at least linear in the spanner size; moreover, the state-of-the-art deterministic algorithm runs in super-linear time and is rather complex \cite{Chazelle00},
and the state-of-the-art linear time algorithms either rely on randomization \cite{KKT95} or on some assumptions, such as the one given by transdichotomous model \cite{FW94}. Instead, using our navigation scheme, as provided by~\Cref{thm:oracle}, we can construct an approximate MST easily, within time $O(n \tau)$, where $\tau$ is the query time as in~\Cref{thm:oracle}. In particular, for low-dimensional Euclidean spaces, we can construct in this way a $(1+\eps)$-approximate MST (for a fixed $\eps$) that is a subgraph of the underlying spanner within $O(n k)$ time, where $k = 2,3,\ldots,O(\alpha(n))$. 

In what follows, we will assume that we are given a Euclidean metric $(X, \delta_X)$, with $|X| \coloneqq n$ and that we have constructed spanner $H_k$ and the data structure $\mathcal{D}_X$ for $(X, \delta_X)$ as in \Cref{thm:oracle}. 
The MST construction is as follows. Initially, compute a $(1+O(\eps))$ minimum spanning tree $T$ on $(X, \delta_X)$ using \cite{chan2008well}. For each edge $(u, v) \in T$, query $\mathcal{D}_X$ for a $k$-hop path $P_{u, v}$ of weight at most $(1+\eps)w(u, v)$. Let $H \coloneqq (X, \cup_{(u, v) \in E(T)}P_{u, v})$ be the union of the obtained paths. Return a spanning tree of $H$.

\begin{theorem}\label{thm:mst}
Given an $n$-point Euclidean space $M_X=(X, \delta_X)$ and its $(1+\eps)$-spanner $H_k$ equipped with a data structure $\mathcal{D}_X$ as in \Cref{thm:oracle}, one can compute a $(1+\eps)$-approximate minimum spanning tree for $(X, \delta_X)$ in $O(n\cdot k/\eps^d)$ time.
\end{theorem}
\begin{proof}
Computing a $(1+O(\eps))$-approximate minimum spanning tree using Chan's algorithm \cite{chan2008well} takes $O(n/\eps^d)$ time. For each edge $(u, v) \in T$ we query $\mathcal{D}_X$ for path $P_{u,v}$ in $H_k$ in time 
$O(k/\eps^d)$ as in \Cref{thm:oracle};
altogether, it takes $O(n\cdot k/\eps^d)$ time to compute $H$.
Since the edges graph $H$ consist of union of $k$-hop paths $P_{u,v}$, the size of $E(H)$ is  $O(n\cdot k)$.
Every path $P_{u, v}$ is within a factor $(1+\eps)$ of $w(u, v)$, so the total weight of $E(H)$ is within a $(1+\eps)$ of the $w(T)$. Computing a spanning tree of $H$ can be done in $O(|E(H)|) = O(n\cdot k)$ time using BFS.
\end{proof}
\begin{remark}
When $\eps$ is constant, the running time in \Cref{thm:mst} becomes $O(n\cdot k)$.
\end{remark}

\subsection{Online tree product and MST verification}\label{sec:treeProduct}

The online tree product problem  \cite{DBLP:journals/jacm/Tarjan79,DBLP:conf/focs/Chazelle84,Alon87optimalpreprocessing, DBLP:journals/combinatorica/Pettie06} is defined as follows.\footnote{This problem is sometimes called the \emph{online tree sum problem}.} Let $T$ be an $n$-vertex tree with each of its edges being associated with an element of a semigroup $(S,\circ)$. One needs to answer online queries of the following form: Given a pair of vertices $u, v \in T$, find the product of the elements associated with the edges along the path from $u$ to $v$. 
A slight variant of this problem is the online MST verification problem where the edge weights of $T$ are real numbers. One needs to answer online queries of the following form: Given a weighted edge $(u,v)$ not in $E(T)$, report if the weight of $(u,v)$ is larger than each edge weight along the path between $u$ and $v$ in $T$.
In both problems, the goal is to design efficient (in terms of time and space) preprocessing and query algorithms which use as few as possible semigroup operations (or binary comparisons).
In \Cref{sec:treeProd}, we show that our navigation algorithm from \Cref{thm:treeNavigate} can be easily modified to support the online tree product problem. Then, we further optimize it for the online MST verification problem in \Cref{sec:mstVerify}.

\subsubsection{Online tree product}\label{sec:treeProd}
Given a tree $T$ and its 1-spanner $G_T$, the algorithm from \Cref{thm:treeNavigate} builds a data structure $\mathcal{D}_T$ such that, for any two vertices $u,v\in V(T)$, it returns a 1-spanner path of at most $k$ hops in $G_T$ in $O(k)$ time. We proceed to show how to preprocess the spanner $G_T$ and data structure $\mathcal{D}_T$ so that each spanner edge $(u,v)\in E(G_T)$ has assigned to it a value from semigroup $(S,\circ)$. This value corresponds to the product of the edge values along the path from $u$ to $v$ in $T$. Using this information assigned to the edges of $G_T$, we can answer online tree product queries between two vertices $u$ and $v$ in $V(T)$ by querying $\mathcal{D}_T$ for a $k$-hop path in $G_T$. This results in online tree product algorithm which uses $k-1$ semigroup operations and $O(k)$ time per query. 

The algorithm we proceed to explain uses the notions introduced in \Cref{SectionPathOracle}. 

We start by describing the preprocessing algorithm for the spanner construction when $k=2$. 
In this case, all the spanner edges are either added from from a cut vertex $u$ to every other required vertex in tree $T$ considered in the current level of recursion, or 
in a base case of the recursive construction (see the description of procedure $\textsc{HandleBaseCase}((T,rt(T)), R(T), k)$ in \Cref{subsec:preprocess-tree}).
For the spanner edges added from cut vertex $u$, we perform DFS traversal on $T$ from $u$ and precompute the semigroup products to every other vertex in $T$. (Note that $(S,\circ)$ might not be commutative and we have to precompute the values both from $u$ to $v$ and from $v$ to $u$.) During the DFS traversal, we use stack to maintain the semigroup product (in both directions) along the path from $u$ to the vertex currently visited. This traversal requires time linear in the number of required vertices in $T$, meaning that it does not asymptotically increase the running time of the preprocessing algorithm (cf. \Cref{lm:preporcessingtime}). The spanner edges added in a base case always shortcut paths of length 2, so we can precompute and store the edge information in the adjacency lists of its endpoints with a constant overhead in time and space per edge. This concludes the description for the case when $k=2$.

When $k=3$, at each each recursion level there is $\Theta(\sqrt{n})$ cut vertices, denoted by $CV_\ell$, where $n$ is the number of required vertices in the tree considered at that level, $T$. The spanner consists of edges 
from each cut vertex $u \in CV_\ell$ to vertices in $\border(u)$,
those in $CV_\ell \times CV_\ell$, and the edges added when the recursion ends in a base case.
To precompute the information associated with edges from each cut vertex $u \in CV_\ell$ to vertices from $\border(u)$, we perform one DFS traversal per cut vertex. For cut vertex $u$, DFS precomputes the semigroup products from $u$ to all the vertices in $\border(u)$, in the same way as when $k=2$. This information is stored together with the edge information in $\mathcal{D}_T$. The running time of DFS from $u$ is linear in the number of spanner edges added from $u$ to vertices in $\border(u)$. 
From Lemma 3.13 in~\cite{Sol13}, we know that the union over all vertices in $CV_\ell$ contains at most $O(n)$ spanner edges of this type, meaning that the total running time spent is linear in $n$.
For the same reason as in the proof of \Cref{lm:preporcessingtime}, this step does not affect the asymptotic running time of the preprocessing algorithm. 
We proceed to explain how to precompute the values associated with the spanner edges between vertices in $CV_\ell \times CV_\ell$. 
First, we set all the vertices in $CV_\ell$ as required and all the other vertices in $T$ as Steiner and invoke pruning procedure $\textsc{Prune}((T,rt(T)),R(T))$ from \Cref{subsec:preprocess-tree} (cf. Section 3.2 in \cite{Sol13}). The output of this procedure is tree $T_\pruned$ which satisfies: (\emph{i}) it has at most $2|CV_\ell|-1$ vertices, and (\emph{ii}) each edge $(x,y) \in E(T_\pruned)$ has associated to it semigroup products in both directions along the path between $x$ and $y$ in $T$. This step requires $O(n)$ time. 
For each cut vertex $u \in CV_\ell$, we perform one DFS on $T_\pruned$ (starting from $u$) and precompute the semigroup products from $u$ to every other vertex in $T_\pruned$. Each DFS call requires size linear in $|T'|$ so the overall complexity of this step is $\Theta(|T_\pruned|^2) = \Theta(n)$. We have described how to preprocess all the semigroup products along the edges in $CV_\ell \times CV_\ell$ within a linear time.
Preprocessing for the edges added in a base case is handled in the same way as when $k=2$ --- for every spanner edge, we store the semigroup product associated with it in the adjacency arrays of its endpoints.

Finally, when $k\ge 4$, the spanner edges are either added from a cut vertex $u$ to vertices in $\border(u)$, or when the recursion reaches a base case (cf. \textsc{HandleBaseCase}$((T,rt(T)), R(T)), k)$ in  \Cref{subsec:preprocess-tree}), or via recursive call using the construction for $k-2$. Precomputing the values associated with the edges in $\{u\} \times \border(u)$ is done in the same way as we described for the case when $k=3$. The edges in the base case are handled in the same way as when $k=2$ and $k=3$. Finally, recall that each composite vertex in $\Phi$ keeps a pointer to a data structure containing recursively precomputed information for $k-2$. We can use this data and recursively precompute the values associated with the spanner edges there.
In summary, we have proved the following theorem.
\begin{theorem}\label{thm:treeSum}
Let $T$ be an $n$-vertex tree having edges associated with elements of semigroup $(S, \circ)$ and let $k\ge 2$ be any integer. We can preprocess $T$ and build a data structure in $O(n\alpha_k(n))$ time and space, such that upon a query for any two vertices $u, v\in V(T)$, it returns the semigroup product along the path from $u$ to $v$ using $k-1$ semigroup operations and in $O(k)$ time.
\end{theorem}

\begin{remark}
For given parameter $k \ge 2$, the online tree product algorithm by \cite{Alon87optimalpreprocessing} achieves preprocessing time $O(n \alpha_k(n))$ but each query follows a path with $2k$ (instead of $k$) hops, thus requiring $2k-1$ semigroup operations.
\end{remark}

Finally, we note that \cite{Alon87optimalpreprocessing} shows several applications of their algorithm due to \cite{DBLP:journals/jacm/Tarjan79}: (\emph{i}) finding maximum flow values in a multiterminal network, (\emph{ii}) verifying minimum spanning trees, and (\emph{iii}) updating a minimum spanning tree after increasing the cost of one of its edges. We can naturally support all these applications using smaller complexity per query while maintaining the same preprocessing time and space guarantees.

\subsubsection{Online MST verification}\label{sec:mstVerify}
We now restrict our attention to the online MST verification problem, i.e., to a variant of the tree sum problem where the edges of the tree are elements of the semigroup $(\mathbb{R}, \max)$. We note that some works considered vertex-weighted versions of this problem, but the two are equivalent up to a linear time transformation which preserves the size of the original tree to within a factor of 2. 
Indeed, if we are given a vertex-weighted tree $T$ with any weight function $w_T : V(T) \to \mathbb{R}$, then we can build an (edge-weighted) tree $T'$ such that $V(T') \coloneqq V(T)$ and $E(T') \coloneqq E(T)$, and assign the weight of an edge $(u,v) \in E(T')$ to be $\max(w_T(u), w_T(v))$. 
Conversely, if we are given an edge-weighted tree $T'$ with any weight function $w_{T'} : E(T') \to \mathbb{R}$, then we can build a (vertex-weighted) tree $T$, where for each edge $(u, v) \in E(T')$, 
we add  to $T$ the two edges $(u, w)$ and $(w, v)$, and set $w_T(u) \coloneqq -\infty$, $w_T(v) \coloneqq -\infty$, $w_T(w) = w_{T'}(u,v)$. The tree $T'$ has $2|V(T)|-1$ vertices and this transformation works in linear time. 
Using these transformations, any algorithm for the edge-weighted variant of the problem can be used for solving the vertex-weighted variant within the same up to constant factors preprocessing and query complexities, and vice versa.
Thus, we may henceforth consider the two variants of the problem as equivalent. 

Koml\'{o}s \cite{DBLP:journals/combinatorica/Komlos85} showed that for any tree $T$ with values associated with its $n$ vertices and a given set of $m$ simple paths (queries) on $T$, one can find the maximum value for each of the $m$ queries using only $O(n\log((m+n)/n))$ comparisons; this algorithm was presented without an (efficient) implementation. 
The first implementation was given by \cite{DBLP:journals/siamcomp/DixonRT92}, and subsequently simpler ones were proposed \cite{DBLP:journals/algorithmica/King97,DBLP:conf/stoc/BuchsbaumKRW98,DBLP:conf/wg/Hagerup09}. These implementations run in time $O(n + m)$ while achieving the same bound $O(n\log((m+n)/n))$ on the number of comparisons during preprocessing. 
We can use this result to reduce the number of comparisons during preprocessing in \Cref{thm:treeNavigate}. Specifically, for any tree $T$, we first construct a $1$-spanner $G_T = (V(T), E)$ and the data structure $\mathcal{D}_T$ as in \Cref{thm:treeNavigate}, and then apply an implementation of Koml\'{o}s' algorithm \cite{DBLP:journals/siamcomp/DixonRT92,DBLP:journals/algorithmica/King97,DBLP:conf/stoc/BuchsbaumKRW98,DBLP:conf/wg/Hagerup09} with the set of queries being (paths between) the endpoints of edges in $E$. We store the precomputed information together with the spanner edges.
As a direct corollary, the navigation scheme from \Cref{thm:treeNavigate} can be precomputed in $O(n\alpha_k(n))$ time and space while using only $O(n \log\alpha_k(n))$ comparisons. 
Given a query edge $(u,v)$, we first query $\mathcal{D}_T$ for a $k$-hop $1$-spanner path between $u$ and $v$. Using $k-1$ comparisons we find the maximum weight along the path and using another comparison we compare this maximum weight against that of the query edge. Overall, this requires $k$ comparisons. The runtime of the query algorithm remains $O(k)$, as in \Cref{thm:treeNavigate}. 

When $k$ is even, the number of comparisons per query can be reduced by one using an idea suggested in \cite{DBLP:journals/combinatorica/Pettie06}. We next describe this idea for completeness. As before, for any $n$-vertex tree $T$, we start by constructing a $1$-spanner $G_T = (V(T), E)$ and the data structure $\mathcal{D}_T$ from \Cref{thm:treeNavigate}. 

Consider first the case $k=2$. We store the $n-1$ edges of $T$ in an array, which we denote by $S$. We sort $S$ using $O(n\log{n})$ comparisons and to each edge $e$ in $T$ we assign a unique integer in $\{1,\dots,n-1\}$, obtained as the position of $e$ in $S$. We call this position the \emph{order} of $e$ in $S$. Next, we assign to each edge $(u,v)$ in $E(G_T)$ a number corresponding to the maximum among the edge \emph{orders} (in $S$) on the path between $u$ and $v$ in $T$. This step can be done in time linear in the spanner size using an implementation of Koml\'{o}s' algorithm. By using the edge orders in the sorted array $S$, rather than their weights, we are not spending any comparisons in this step. Given a query edge $(u,v)$, we first query $\mathcal{D}_T$ for a $2$-hop $1$-spanner path between $u$ and $v$. Let the two edges of this path be $(u,w)$ and $(w,v)$. Without using any weight comparison, we can find which of the two edges have larger order associated to them. We then find the edge $e$ in $S$ having this order and using one comparison compare the weight of $e$ to the weight of the query edge, $(u,v)$. Overall the number of comparisons made is $k-1=1$.

Suppose now that $k\ge2$ is an even integer. Whenever \textsc{PreprocessTree}$((T,rt(T)), R(T),k)$ for constructing $\mathcal{D}_T$ executes a recursive call for some tree $T'$ with parameter $k=2$, we preprocess the spanner edges added for $T'$ using what we explained in the previous paragraph for the case $k=2$. Using a recurrence similar to that in \Cref{lm:preporcessingtime} (see also the proof of Theorem 3.12 in \cite{Sol13}), it is easy to verify that for $k>2$, the total number of comparisons over all the trees considered in recursive calls with parameter $k=2$ is linear in the size of given tree $T$.
For the other spanner edges in $E(G_T)$, we apply an implementation of Koml\'{o}s' algorithm and store the precomputed information alongside the edges. This step requires $O(n \alpha_k(n))$ time and space and uses $O(n\log \alpha_k(n))$ comparisons. Thus, the total time and space complexity for preprocessing is $O(n \alpha_k(n))$ and the number of comparisons used is $O(n \log \alpha_k(n))$. 
Given a query edge $(u,v)$, we first query $\mathcal{D}_T$ for a $k$-hop $1$-spanner path between $u$ and $v$. From procedure \textsc{FindPath}$(u, v, \Phi, k)$, we know that this path either contains less than $k$ edges or it contains $k$ edges, two edges of which, $e_1, e_2$, belong to some tree $T'$ that was preprocessed with parameter $k=2$. (See line \ref{line:query:k2} in procedure $\Call{FindPath}$; the edges $e_1$ and $e_2$ correspond to $(u, \phi_{T'}(\beta))$ and $(\phi_{T'}(\beta), v)$.)
In the former case, the number of comparisons required is clearly at most $k-1$.
If the latter case, i.e., the number of edges on the spanner path between $u$ and $v$ is equal to $k$, we save one comparison by comparing the \emph{orders} of edges $e_1$ and $e_2$ in $S$ for $T'$, which reduces the number of comparisons from $k$ to $k-1$.

The following theorem summarizes the guarantees of our online MST verification algorithm.

\begin{theorem}\label{thm:verifyMST}
Let $T$ be an edge-weighted tree with $n$ vertices and let $k\ge 2$ be any integer. We can preprocess $T$ and build a data structure in $O(n\alpha_k(n))$ time and space and using $O(n\log \alpha_k(n))$ comparisons, such that it answers online MST verification queries on $T$ in $O(k)$ time and using at most $k-1$ comparisons.
\end{theorem}

\begin{remark}
Alon and Schieber \cite{Alon87optimalpreprocessing} gave an algorithm for the online MST verification problem that requires $O(n \alpha_k(n))$ time, space and comparisons during preprocessing. Their algorithm answers queries following paths of length $2k$, thus achieving $2k$ comparisons. (The number of comparisons is $2k$, rather than $2k-1$, since the query edge must be compared against the tree weights.) Our algorithm improves this tradeoff both in terms of the number of comparisons required for preprocessing ($O(n\log\alpha_k(n))$
rather than $O(n \alpha_k(n))$) and the number of comparisons per query ($k-1$ rather than $2k-1$).
Pettie \cite{DBLP:journals/combinatorica/Pettie06} shows that it suffices to spend $O(n \alpha_{2k}(n))$ time and space and $O(n \log \alpha_{2k}(n))$ comparisons during preprocessing, so that each subsequent query can be answered using $4k-1$ comparisons. In fact, \cite{DBLP:journals/combinatorica/Pettie06} uses a different variant of a row-inverse Ackermann function, $\lambda_k$, which satisfies $\lambda_k(n) = \Theta(\alpha_{2k}(n))$ (see \Cref{lemma:compareAckermann}). This result of \cite{DBLP:journals/combinatorica/Pettie06} builds on  \cite{Alon87optimalpreprocessing} (and \cite{DBLP:conf/focs/Chazelle84}), which requires $4k$ comparisons (rather than $2k$ claimed in \cite{DBLP:journals/combinatorica/Pettie06}) following a preprocessing time of $O(n\lambda_k(n))=O(n \alpha_{2k}(n))$. To reduce the resources required for preprocessing,  Pettie \cite{DBLP:journals/combinatorica/Pettie06} used Koml\'{o}s' algorithm \cite{DBLP:journals/combinatorica/Komlos85}, which reduces the number of \emph{comparisons} to $O(n \log \alpha_{2k}(n))$, but not the running time, since the algorithm is information-theoretic, and all known implementations of Koml\'{o}s' algorithm take \emph{time} linear in the number of queries (which is $\Theta(n\alpha_{2k}(n))$ in this case) \cite{DBLP:journals/siamcomp/DixonRT92,DBLP:journals/algorithmica/King97,DBLP:conf/stoc/BuchsbaumKRW98,DBLP:conf/wg/Hagerup09}.
In this regime, our algorithm requires $O(n  \alpha_{2k}(n))$ time and space and $O(n \log \alpha_{2k}(n))$ comparisons during preprocessing, so that each subsequent query can be answered using $2k-1$ comparisons in $O(k)$ time. The result of \cite{DBLP:journals/combinatorica/Pettie06} can also achieve a query time of $O(k)$ (though not claimed), by building on \cite{Alon87optimalpreprocessing}, but using $4k-1$ comparisons rather than $2k-1$ as in our result.
Concurrently and independently of us, Yang \cite{DBLP:journals/corr/abs-2105-01864} obtained a result similar to ours;
the two techniques are inherently different.
Interestingly, Yang \cite{DBLP:journals/corr/abs-2105-01864}
defines yet another variant of a row-inverse Ackermann function, under the notation $\lambda_k$, which is similar to the function used by Pettie \cite{DBLP:journals/combinatorica/Pettie06}.  
\end{remark}

\section*{Acknowledgments}
The fourth-named author thanks Ofer Neiman for helpful discussions.

\bibliographystyle{alpha}

\bibliography{refs,ENS14,bib,ENS14V2,spanner,spanner2}
\appendix
\section{Proof of Lemma~\ref{lemma:compareAckermann}}\label{sec:proofAckermann}
This \namecref{sec:proofAckermann} is dedicated to proving the following lemma stated in \Cref{sec:prelim}.
\compareAckermann*
Let $T(\cdot,\cdot)$ be slightly different Ackermann function as defined by Tarjan \cite{DBLP:journals/jacm/Tarjan75}.
\begin{align*}
T(0, j) &= 2j & \text{ for } j \ge 0\\
T(i, 0) &= 0 & \text{ for } i \ge 1\\
T(i, 1) &= 2 & \text{ for } i \ge 1\\
T(i, j) &= T(i-1, T(i, j-1)) & \text{ for } i\ge 1 \text{ and } j\ge 2
\end{align*}

In the following lemma, we show that $A$ and $T$ are almost equal (except for the first column).
\begin{claim}
For all $i\ge 0$ and $j\ge 1$, $A(i, j) = T(i, j)$.
\end{claim}
\begin{proof} We will show the claim inductively.\\
\noindent
\emph{Base case $i=0$.}
For all $j\ge 0$, it follows by definition that $A(0, j) = T(0,j) = 2j$. 

\noindent
\emph{Base case $j=1$.}
For any $i > 0$,it follows by definition that $A(i, 1) = A(i-1, A(i, 0)) = A(i-1, 1)$. Since $A(0,1)=2$ it follows by induction that for every $i\ge 0$, $A(i,1)=2$, and so $A(i, 1) = T(i,1)$. 

\noindent
\emph{Inductive step.}
We proceed to show the inductive step for $i\ge 1, j\ge 2$, assuming that the claim holds for any pair $i' \ge 0, j' \ge 1$ lexicographically smaller than $(i, j)$.
\begin{align*}
A(i,j) &= A(i-1, A(i, j-1))\\
&= A(i-1, T(i, j-1)) &\text{ from inductive hypothesis }\\
&= T(i-1, T(i, j-1)) &\text{ since } T(i,j-1) \ge 2 \text{ for } j \ge 2\\
&= T(i,j)
\end{align*}
\end{proof}
We are ready to prove \Cref{lemma:compareAckermann}.
\begin{proof}[Proof of \Cref{lemma:compareAckermann}]
Pettie \cite{DBLP:journals/combinatorica/Pettie06} shows that for any $i\ge 1$ his variant of Ackermann function $P$ (cf. \Cref{sec:ackermann}) satisfies:
\begin{itemize}
\item $T(i,j)\le P(i,j)$ for $j\ge 0$,
\item $P(i,j) \le T(i, 3j)$ for $j\ge 1$.
\end{itemize}

From the definition of $\lambda_i(\cdot)$:
\begin{align*}
\lambda_i(n) &= \min\{j : P(i,j) \ge n\}\\
&\le \min\{j : T(i,j) \ge n\} \\
&= \min\{j : A(i,j) \ge n\} &\text{ for } j \ge 1\\
&= \alpha_{2i}(n).
\end{align*}

On the other hand:
\begin{align*}
\lambda_i(n) &= \min\{j : P(i,j) \ge n\}\\
&\ge \min\{j : T(i,3j) \ge n\} &\text{ for } j \ge 1\\
&= \min\{j : A(i,3j) \ge n\} \\
&\ge \frac{1}{3}  \alpha_{2i}(n)
\end{align*}
\end{proof}

\end{document}